\documentclass{article}
\usepackage{graphicx} 
\usepackage{amsmath}
\usepackage{amsthm}
\usepackage{amssymb}
\usepackage{mathtools}
\usepackage{xcolor}
\usepackage[round]{natbib}
\usepackage{enumerate}
\usepackage{bbm}
\usepackage{bm}
\usepackage{physics}
\usepackage{apptools}
\usepackage{float}
\usepackage{threeparttable}
\usepackage{rotating} 
\usepackage{enumitem}
\usepackage{subfig}
\usepackage{adjustbox}

\usepackage[a4paper, total={6in, 8in}]{geometry}

\usepackage{authblk}

\usepackage{orcidlink}

\AtAppendix{\counterwithin{theorem}{section}}

\newtheorem{theorem}{Theorem}
\newtheorem{lemma}[theorem]{Lemma}
\newtheorem{remark}[theorem]{Remark}

\theoremstyle{definition}

\providecommand{\keywords}[1]
{
  \small	
  \textbf{\textit{Keywords:}} #1
}

\newcommand\extrafootertext[1]{%
    \bgroup
    \renewcommand\thefootnote{\fnsymbol{footnote}}%
    \renewcommand\thempfootnote{\fnsymbol{mpfootnote}}%
    \footnotetext[0]{#1}%
    \egroup
}

\title{A Better Comparison under right-censoring: ABC Statistic for Equivalence Testing and Quantification}

\author{
Simon Mack $^{1,\ast}$\orcidlink{0009-0006-0081-935X},
Kathrin Möllenhoff $^{2}$\orcidlink{0000-0001-7861-3892} and
Dennis Dobler $^{1}$\orcidlink{0000-0002-9040-0854}}

\affil{
$^{1}$ Institute of Statistics, RWTH Aachen University, Aachen, Germany \\
$^{2}$ Institute of Medical Statistics and Computational Biology, University of Cologne, Cologne, Germany}

\date{\today}

\begin{document}

\maketitle

\begin{abstract}
The ABC (area between curves) statistic is an $L^1$-distance which targets an easy-to-interpret estimand. 
Defined as the (normalized) integrated absolute distance between two survival curves it is a meaningful quantity even when survival functions are crossing.
Based on right-censored time-to-event data, estimation is based on Kaplan-Meier curves obtained from two independent sample groups.
In the present paper, we develop the large sample properties of the ABC statistic and investigate various resampling options for approximating the statistic's distribution which is possibly non-normal in the limit.
These breakthroughs enable the construction of equivalence tests which can be used to establish that differences between two survival functions are practically irrelevant.
Alternatively, the point estimator can be accompanied with confidence intervals that comprehensibly quantify the difference between the curves.
An extensive simulation study explores these inferential methods under various scenarios: proportional, crossing, and partially equal survival functions.
An application to data on overall and progression-free survival in a lung cancer trial illustrates the methods' benefits and some points of consideration.
\end{abstract}

\extrafootertext{$\ast$ Corresponding author. Email adress: \texttt{simon.mack@rwth-aachen.de}}
\section{Introduction}

Time-to-event outcomes are central in medical research, particularly in areas such as oncology and cardiovascular disease, where understanding differences in survival between treatments is essential for clinical decision-making. Typically, the aim is to test whether there is a significant difference between two treatment groups. However, depending on the goal of a study, researchers are sometimes more interested in claiming that there is no \textit{relevant} difference between the two groups. To this end, instead of testing an alternative hypothesis of inequality, the goal is to show that they do not differ substantially, i.e. less than some pre-specified margin $d$. This equivalence testing approach provides, due to the usually free choice of the margin $d$, a very flexible framework for many research questions (see, for example, \citealp{wellek2010testing}). Depending on the study, a non-inferiority test may be used instead of an equivalence test, e.g., to show that a new treatment is not inferior to another treatment. Both tests, non-inferiority and equivalence, require the same statistical methods, which are different from those used to test classical difference hypotheses. 

In general, classical methods widely used to compare survival outcomes across treatment groups include Kaplan–Meier estimates combined with log-rank tests, and, when adjusting for covariates, the Cox proportional hazards model. These approaches have become standard because they do not assume a specific parametric form for the event time distribution and because the hazard ratio provides a single summary measure of treatment effect. Specifically, also when the goal of the study is not to demonstrate a significant difference, but rather to demonstrate equivalence or non-inferiority, approaches are typically based on these classical methods by using modifications of the log-rank test and thus also depend on the assumption of proportional hazards \citep{wellek1993}. 
It is well known that this assumption is crucial and frequently violated (see, e.g., \citealp{hernan2010,uno2014,kuemmerli2021}) and, if so, such approaches suffer from a considerable loss of power \citep{lin2020}. 
Therefore, alternative measures and tests have been proposed and compared. These include, for example, weighted log-rank tests \citep{fleming2013counting}, combination tests \citep{breslow1984,logan2008} and tests based on the difference or ratios of the RMST \citep{royston2013,trinquart2016,munko2024rmst}. 

However, all these previously mentioned tests aim for testing a significant difference between groups rather than claiming equivalence.
To address this problem, non-parametric methods investigating equivalence have been introduced based on the difference of Kaplan-Meier curves \citep{com1993}, the odds ratio \citep{martinez2017}, as well as semi-parametric methods based on a log transformation model \citep{shen2020}. In a work of \citealp{zhao2016}, such a test based on the RMST has been proposed by constructing simultaneous confidence bands for the difference between two RMST curves. Further, \citealp{mollenhoff2023}, have proposed a parametric approach to assess equivalence, in which survival curves are fitted according to various distributions of event times, and afterwards, confidence intervals for both the difference in survival curves and the hazard ratio are derived asymptotically and via bootstrapping methods. 

More generally, when it comes to equivalence testing of functions, several other measures have been considered, for example the $L^2$-distance \citep{detmolvolbre2015} and, extensively, the maximum distance between curves \citep{detmolvolbre2015,mollenhoff2023,hagemann2025}. In a recent work of \citealp{bastian2024comparing}, the $L^1$-distance was introduced as a measure of equivalence, focusing on (parametric) regression curves. In the two-sample survival setting the $L^1$-distance has been used already by \citet{lin2010new} to test the equality of survival functions, however, without establishing the correct asymptotic distribution which was found by \citet{liu2020resampling}.

In the present paper we extend the methodology in two directions: first, we focus on equivalence testing instead of classical hypotheses of inequality which is still a gap in the survival literature. Second, instead frequently used quantities such as the hazard ratio or the difference in RMST, we consider the $L^1$-distance as a measure of equivalence.
This bears the advantage of robustness because it does not rely on restrictive model assumptions, in particular, not on the proportional hazards assumption. Compared to the RMST, the $L^1$-distance further has the advantage that it can characterize the equivalence of survival functions; in contrast, two RMSTs can be equal although the underlying survival curved are not necessarily the same. 

This paper is organized as follows: First, we propose a modified $L^1$-distance as an estimator to describe the discrepancy between two survival functions and investigate its asymptotic properties. Secondly, we derive an equivalence test based on this estimator, using two different resampling-based test procedures. Next, we present extensive simulation results to demonstrate the finite sample performance of our methods. Finally, we apply our approach to the METLung trial \citep{spigel2017}, evaluating onartuzumab added to standard chemotherapy in advanced non-small cell lung cancer (NSCLC).

\section{Notation}
We consider a two-sample survival setting with $n_j \in \mathbb{N}, j=1,2$, observations of group-wise i.i.d.\ tupels $(X_{ij}, \delta_{ij}), i=1,\dots,n_j$, with observable times $X_{ij}=\min(T_{ij}, C_{ij})$ and survival functions
$T_{ij}\sim S_j$ and $C_{ij}\sim G_j$
representing the event and censoring times, respectively, and $\delta_{ij}=\mathbbm{1}\{T_{ij}\leq C_{ij}\}$ indicates the censoring status. We assume the nonnegativity and  stochastic independence (``random right-censoring'') of event and censoring times, but do not impose any additional model assumptions. In particular, $S_1$ and $S_2$ need not be continuous. Based on the $n_1+n_2=n$ pooled observations, we wish to decide whether the two survival functions $S_1$ and $S_2$ are ``equivalent'' on the interval $[0,\tau]$ where the terminal time $\tau$ has to be chosen such that $P(X_{ij}\geq \tau)>0.$ In a clinical trial, $\tau$ usually represents the end of study time. Unlike with finite dimensional parameters, the concept of \textit{distance} between survival functions is somewhat ambiguous and any metric on the space of probability distributions could be used to construct a notion of equivalence. We propose a modified $L^1$-distance,

\begin{equation}\label{theoreticalL1_distance}
    \Delta_{\tau}(S_1, S_2)= \frac{1}{\tau}\int_0^{\tau}|S_1(t)-S_2(t)|dt,
\end{equation}
to describe the discrepancy between two survival functions.
It is geometrically intuitive as the standardized 
area between both curves and, unlike the popular Kolmogorov-Smirnov distance, takes all function values on $[0,\tau]$ into account instead of only the maximum distance. The factor $1/\tau$ ensures that the distance does not depend on the used time scale and is contained in $[0,1].$ Based on this distance, we can formulate our equivalence testing problem:
\begin{equation}\label{equivalencehypothesis}
    H_0:\Delta_{\tau}(S_1, S_2) \geq \varepsilon \quad \text{versus}\quad H_1:\Delta_{\tau}(S_1, S_2)<\varepsilon
\end{equation}
for some pre-specified threshold value $\varepsilon \in (0,1).$

Adopting the standard counting process notation of  \citet{andersen1993statistical}, we denote for $t\geq0$ by $N_j(t)=\sum_{i=1}^{n_j}\mathbbm{1}\{X_{ij}\leq t, \delta_{ij}=1\}$ the number of observed events up until $t$ in group $j$ and by $Y_j(t)=\sum_{i=1}^{n_j}\mathbbm{1}\{X_{ij}\geq t\}$ the number of individuals at risk just before $t.$ In terms of these processes the group-wise Kaplan-Meier estimators are given by
    $\hat{S}_j(t)=\prod_{s\leq t}\left(1-\frac{\Delta N_j(s)}{Y(s)}\right).$
Here $\Delta N_j(s)=N_j(s)-N_j(s-)$ counts the number of events happening at time $s$ and $N_j(s-)$ denotes the left-continuous version of $N_j.$

\section{Asymptotic Inference}
We immediately obtain a point estimator for our discrepancy measure \eqref{theoreticalL1_distance} via plug-in:

\begin{equation}\label{estimatedL1_distance}
    \hat{\Delta}_{\tau}(S_1, S_2)= \frac{1}{\tau}\int_0^{\tau}|\hat{S}_1(t)-\hat{S}_2(t)|dt.
\end{equation}
As the Kaplan-Meier estimators are step functions, this estimator reduces to a sum:  Denote by $0 \leq t_1 < \dots < t_{k_n} \leq \tau$ the ordered distinct relevant event times in the pooled sample $X_{11},\dots,X_{n_11},$ $X_{12},\dots, X_{n_22}$; additionally, set $t_{k_n+1}=\tau.$ Then \eqref{estimatedL1_distance} becomes
\begin{equation*}
    \hat{\Delta}_{\tau}(S_1, S_2)=\frac{1}{\tau}\sum_{i=1}^{k_n}|\hat{S}_1(t_i)-\hat{S}_2(t_i)|(t_{i+1}-t_i).
\end{equation*}
This estimator is also useful outside the framework of testing for equivalence. Standard two-sample tests like the log-rank test lack the possibility to quantify the difference between survival functions in an interpretable manner if the null hypothesis of equality is rejected. In these scenarios $ \hat{\Delta}_{\tau}(S_1, S_2)$ serves as an interpretable effect measure for describing the difference of survival functions and could help to assess the practical significance of the results. The following theorem demonstrates the consistency of the estimator $\hat{\Delta}_{\tau}(S_1, S_2)$.

\begin{theorem}\label{consistency_effectmeasure}
    As $\min(n_1,n_2)\to \infty,$ the following convergence holds (outer) almost surely:
    $$\hat{\Delta}_{\tau}(S_1, S_2)\to \Delta_{\tau}(S_1, S_2).$$
\end{theorem}
For constructing a test for equivalence, we need to analyze the limiting distribution of the estimator. To this end, we impose that no (relative) group size vanishes asymptotically:

    \begin{equation} \label{groupassumption}
    \frac{n_1}{n} \to \kappa^{(1)} \in (0,1),\quad \frac{n_2}{n} \to \kappa^{(2)} = 1 - \kappa^{(1)} \quad \text{as}\quad \min(n_1,n_2)\to \infty.
\end{equation}
Subsequently, we denote by $\xrightarrow{d}$ convergence in distribution.
\begin{theorem}\label{limitABCgeneral}
    Suppose assumption \eqref{groupassumption} holds and let $W_1, W_2$ be two independent standard Wiener processes. Then we have, as $\min(n_1,n_2)\to \infty$,
    \begin{equation*}
    \begin{split}
        \sqrt{n}(\hat{\Delta}_{\tau}(S_1, S_2)-\Delta_{\tau}(S_1, S_2))\xrightarrow{d} &D_{\Delta_{\tau}}\coloneqq\frac{1}{\tau}\int_{\{t\in[0,\tau]:S_1(t)=S_2(t) \}}|\mathbb{G}(t)|dt\\
                    &+
        \frac{1}{\tau}\int_{\{t\in[0,\tau]:S_1(t)\ne S_2(t)\}}\operatorname{sign}(S_1(t)-S_2(t))\mathbb{G}(t)dt
    \end{split}
    \end{equation*}
    with $\mathbb{G}(t)=[S_1(t)W_1(\Gamma_1(t)) + S_2(t)W_2(\Gamma_2(t))]$
    and $\Gamma_j(t)=\frac{1}{\kappa_j}\int_0^t\frac{1}{(1-\Delta A_j(s))S_j(s-)G_j(s-)}dA_j(s).$
\end{theorem}
Here $A_j(t)=-\int_0^t\frac{dS_j(s)}{S_j(s-)}$ denotes the group-specific cumulative hazard function, and $\operatorname{sign}(x)=1\{x>0\}-1\{x<0\}$ is the signum function. This limiting distribution is a mixture of a non-Gaussian and a Gaussian random variable, with in general unknown mixture weights $m_1(t)=\frac1\tau\mathbbm{1}\{S_1(t)= S_2(t)\}$ and $m_2(t)=\frac1\tau\mathbbm{1}\{S_1(t)\neq S_2(t)\}$. If one is willing to make the assumption that $S_1$ and $S_2$ only coincide on a set of Lebesgue measure zero, the first summand of the limiting random variable vanishes. In this case $D_{\Delta_{\tau}}$ is Gaussian, and drawing inferences would simplify considerably. However, we do not make this assumption as it is unrealistic in many medical applications; furthermore, it can not be statistically tested. 

The limiting distribution of $D_{\Delta_{\tau}}$ can not directly be used for developing inference procedures, as it contains unknown quantities, i.e., the unknown survival functions of event and censoring distributions as well as the mixture weights. Nonetheless, denoting by $q_{\alpha}$ the $\alpha$-quantile of $D_{\Delta_{\tau}}$ we can define the one sided ``oracle'' confidence interval
\begin{equation*}\label{orableinterval}
    \left[0,\hat{\Delta}_{\tau}(S_1, S_2)-\frac{q_{\alpha}}{\sqrt{n}}\right)
\end{equation*}
which contains $\Delta_{\tau}(S_1, S_2)$ with probability converging to $1-\alpha.$ Based on this interval, we obtain a test for our hypothesis of equivalence \eqref{equivalencehypothesis}:
\begin{equation}\label{oracletest}
    \varphi_{n,\alpha}=\mathbbm{1}\left\{\hat{\Delta}_{\tau}(S_1, S_2)-\varepsilon\leq\frac{q_{\alpha}}{\sqrt{n}} \right\}.
\end{equation}
The following theorem establishes that \eqref{oracletest} defines a consistent and asymptotic level $\alpha$ test.

\begin{theorem}\label{oracletest_limts}
    Suppose assumption \eqref{groupassumption} holds. Then as $\min(n_1,n_2)\to \infty$
    \begin{enumerate}
        \item[(a)] If $\Delta_{\tau}= \varepsilon$, then $\lim_{n\to\infty} E(\varphi_{n,\alpha})=\alpha.$
        \item[(b)] If $\Delta_{\tau}> \varepsilon$, then $\lim_{n\to\infty} E(\varphi_{n,\alpha})=0.$
        \item[(c)] If $\Delta_{\tau}< \varepsilon$, then $\lim_{n\to\infty} E(\varphi_{n,\alpha})=1.$
    \end{enumerate}
\end{theorem}
\begin{remark}
    Based on Theorem \ref{limitABCgeneral} two- and left-sided confidence intervals for $\Delta_{\tau}(S_1, S_2)$ can be obtained as well. By using the latter, a test for the goodness-of-fit hypothesis of equal survival functions can be defined, which leads to the test described in \citet{liu2020resampling}.
\end{remark}
In order to obtain a feasible hypothesis test, we need methods to approximate the unknown quantile $q_{\alpha}$. We propose resampling-based approaches, as explained in the next section.

\section{Resampling-based tests}
\subsection{Subsampling}\label{sec:subsample}
A very general resampling approach is the subsampling procedure introduced originally for one-sample statistics by \citet{politis1994large}. The general idea is to recompute the statistic of interest on a subsample drawn from the original observations without replacement, and to use the resulting discrete distribution as an approximation. We use a general multi-sample extension described in \citet{politis2010k}: we choose integers $1\leq b_j\leq n_j$, representing the size of our subsamples, which fulfill the following convergence rate condition:
\begin{equation}\label{subsamplingassumption}
    b_j/n_j\to0 \quad\text{as}\quad n_j\to \infty \quad\text{with}\quad b_j\to\infty.
\end{equation}
 Let $Z_{1j},\dots,Z_{N_j}$ be equal to the $N_j=\binom{n_j}{b_j}$ possible unordered subsets of $(X_{1j}, \delta_{1j}), \dots, (X_{n_jj}, \delta_{n_jj}).$ Denote by $\hat{S}_{l_jj}$ the Kaplan-Meier estimator computed from the subset $Z_{l_j}.$ Then the subsampling distribution of $\hat{\Delta}_{\tau}(S_1, S_2)$ is defined as
\begin{equation*}
    L_{n,b}(x)=\frac{1}{\binom{n_1}{b_1}\binom{n_2}{b_2}}\sum_{l_1=1}^{\binom{n_1}{b_1}}\sum_{l_2=1}^{\binom{n_2}{b_2}}\mathbbm{1}\left\{\sqrt{b}\left(\frac{1}{\tau}\int_0^{\tau}|\hat{S}_{l_11}(t)-\hat{S}_{l_22}(t)|dt-\frac{1}{\tau}\int_0^{\tau}|\hat{S}_1(t)-\hat{S}_2(t)|dt\right)\leq x\right\}
\end{equation*}
with $b=b_1+b_2.$ By Theorem 3.1 in \citet{politis2010k} $L_{n,b}$ is a consistent estimator for the distribution of $D_{\Delta_{\tau}},$ which leads to the one-sided asymptotic $1-\alpha$ confidence interval
\begin{equation*}
    \left[0,\hat{\Delta}_{\tau}(S_1, S_2)-\frac{L^{-1}_{n,b}(\alpha)}{\sqrt{n}}\right)
\end{equation*}
where $L^{-1}_{n,b}(\alpha)$ is the $\alpha$-quantile of $L_{n,b}$, considered as a (random) distribution function. The corresponding hypothesis test for the hypothesis of equivalence is given by
\begin{equation*}\label{subsampletest}
    \varphi_{n,b,\alpha}^S=\mathbbm{1}\left\{\hat{\Delta}_{\tau}(S_1, S_2)-\varepsilon\leq\frac{L^{-1}_{n,b}(\alpha)}{\sqrt{n}} \right\}.
\end{equation*}
This test has the same desirable properties as the oracle test $\varphi_{n,\alpha},$ as demonstrated here:
\begin{theorem}\label{subsampletest_limits}
    Suppose assumptions \eqref{groupassumption} and \eqref{subsamplingassumption} hold. Then as $\min(n_1,n_2)\to \infty$
    \begin{enumerate}
        \item[(a)] If $\Delta_{\tau}= \varepsilon$, then $\lim_{n\to\infty} E(\varphi_{n,b,\alpha}^S)=\alpha.$
        \item[(b)] If $\Delta_{\tau}> \varepsilon$, then $\lim_{n\to\infty} E(\varphi_{n,b,\alpha}^S)=0.$
        \item[(c)] If $\Delta_{\tau}< \varepsilon$, then $\lim_{n\to\infty} E(\varphi_{n,b,\alpha}^S)=1.$
    \end{enumerate}
\end{theorem}
A natural choice for the subsampling rates would be $b_j=n_j^{\gamma}$ for $\gamma\in (0,1)$ and the question arises, whether an optimal $\gamma$ exists. To the best of our knowledge no definitive answer has been found to date, but heuristic arguments of \citet{politis1994large} suggest that the choice $b_j=n_j^{2/3}$ is first-order optimal in many situations.

\subsection{Modified bootstrap}\label{sec:modified_boot}
As an alternative to the previous subsampling approach, we consider a modified bootstrap approach of \citet{fang2019inference} based on a generalization of the Delta method to functionals that are only \textit{directionally} Hadamard-differentiable. To elaborate, denote by $D[0,\tau]$ the space of càdlàg functions which we equip with the supremum norm and define the mapping $\Psi(f)=\frac{1}{\tau}\int_0^{\tau}|f(t)|dt$ from $D[0,\tau]$ to $\mathbb{R}.$ It is demonstrated in the appendix that $\Psi$ is directionally Hadamard-differentiable, cf.\ Definition~2.2 in \citet{shapiro1990concepts}, with derivative 
$$\Psi'_{f}(h)=\frac{1}{\tau}\int_{\{f=0\}}|h(x)|dx+\frac{1}{\tau}\int_{\{f\ne0\}}\operatorname{sign}(f(x))h(x)dx$$
and $f,h\in D[0,\tau].$ 
If the derivative $\Psi'_{f}(h)$ would be linear, which by Proposition 2.1 of \citet{fang2019inference} is equivalent to ordinary Hadamard differentiability of $\Psi$, the Delta method for the bootstrap (Theorem 3.10.11 of \citet{vanderVaartWellner2023}) would imply that the distribution of $D_{\Delta_{\tau}}$ can be consistently estimated by the conditional distribution of
\begin{equation}\label{naive_boot}
   \sqrt{n}(\Psi(\hat{S}_1^{(B)}-\hat{S}_2^{(B)})-\Psi(\hat{S}_1-\hat{S}_2))
\end{equation}
given the data. Here $\hat{S}_j^{(B)}$ denotes the Kaplan-Meier estimator computed from a bootstrap sample
\\
$ (X^{(B)}_{1j}, \delta^{(B)}_{1j}),$ $\dots, (X^{(B)}_{n_jj}, \delta^{(B)}_{n_jj})$ which is obtained by drawing group-wise with replacement from the original sample. As the process $\mathbb{G}$ is Gaussian, it follows by an application of Corollary 3.1 of \citet{fang2019inference}, that the bootstrap fails in general in our case, because $\Psi'_{S_1-S_2}(h)$ is linear if and only if $S_1$ and $S_2$ coincide on a set of Lebesgue measure zero. As we do not make this assumption in advance, bootstrap validity depends on the unknown event time distributions, which is undesirable.

As a consistent alternative, \citet{fang2019inference} propose a modified bootstrap procedure, based on an estimator for the directional derivative. This approach was also used by \citet{bastian2024comparing} in the context of comparing regression functions in terms of the $L^1$ distance. Adapting their estimator to the scaled distance $\Delta_{\tau}(S_1,S_2),$ we define
\begin{equation}\label{estimated_derivative}
    \hat{\Psi}_n'(h)=\frac{1}{\tau}\int_{\{t\in[0,\tau]:|\hat{S}_1(t)-\hat{S}_2(t)|\leq \frac{1}{c_n} \}}|h(t)|dt +\frac{1}{\tau}\int_{\{t\in[0,\tau]:|\hat{S}_1(t)-\hat{S}_2(t)|> \frac{1}{c_n} \}}\operatorname{sign}(\hat{S}_1(t)-\hat{S}_2(t))h(t)dt
\end{equation}
for a sequence of constants $c_n$ which also fulfill a rate condition:
\begin{equation}\label{boot_ratecondition}
    c_n\to\infty \quad\text{with}\quad c_n/\sqrt{n}\to0 \quad\text{as}\quad \min(n_1,n_2)\to\infty.
\end{equation}
Based on this estimator we obtain the modified bootstrap approximation
\begin{equation}\label{modboot-approximation}
    \hat{\Psi}'_n(\sqrt{n}((\hat{S}_1^{(B)}-\hat{S}_2^{(B)}) -(\hat{S}_1-\hat{S}_2)))
\end{equation}
which consistently estimates the distribution of $D_{\Delta_{\tau}}$ given the data. If we denote by $q^{B}_{n,\alpha}$ the conditional $\alpha$-quantile of \eqref{modboot-approximation}, we obtain the equivalence test
\begin{equation*}
        \varphi_{n,\alpha}^B=\mathbbm{1}\left\{\hat{\Delta}_{\tau}(S_1, S_2)-\varepsilon\leq\frac{q^{B}_{n,\alpha}}{\sqrt{n}} \right\},
\end{equation*}
which again shares its properties with the oracle test \eqref{oracletest}.
\begin{theorem}\label{boottest_limits}
    Suppose assumptions \eqref{groupassumption} and \eqref{boot_ratecondition} hold. Then as $\min(n_1,n_2)\to \infty$
    \begin{enumerate}
        \item[(a)] If $\Delta_{\tau}= \varepsilon$, then $\lim_{n\to\infty} E(\varphi_{n,\alpha}^B)=\alpha.$
        \item[(b)] If $\Delta_{\tau}>\varepsilon$, then $\lim_{n\to\infty} E(\varphi_{n,\alpha}^B)=0$
        \item[(c)] If $\Delta_{\tau}< \varepsilon$, then $\lim_{n\to\infty} E(\varphi_{n,\alpha}^B)=1.$
    \end{enumerate}
\end{theorem}

\begin{remark}
    The estimator $\hat{\Psi}_n'(h)$ can be replaced by the numerical derivative
    \begin{equation}\label{num_diff}
        \Tilde{\Psi}_1(S_1,S_2)=\frac{\Psi[(\hat{S}_1-\hat{S}_2)+\varepsilon_n\sqrt{n}((\hat{S}_1^{(B)}-\hat{S}_2^{(B)}) -(\hat{S}_1-\hat{S}_2))]- \Psi(\hat{S}_1-\hat{S}_2)}{\varepsilon_n}
    \end{equation}
    as outlined by \citet{hong2018numerical}. Here $(\varepsilon_n^{-1})_{n \in \mathbb{N}}$ fulfills the same assumptions as $(c_n)_{n\in \mathbb{N}}.$
\end{remark}
\subsection{Approximation errors and possible remedies}\label{sec:corrections_typeI_errors}
The previous results were purely asymptotic. Therefore, the question arises how well the resampling procedures perform for ``realistic''  sample sizes, especially regarding type-I error control of hypothesis tests. In clinical trails the number of observations per treatment group is usually in the low to mid hundreds, which facilitates that such a consideration is essential in order to apply the proposed methods in practice.
We focus on the following two aspects:
\begin{enumerate}
    \item[(i)] Convergence speeds of the Type-I error rate,
    \item[(ii)] direction of the bias in the finite sample Type-I error.
\end{enumerate}
We treat (i) only somewhat heuristically, as a rigorous analysis would require formal expansions of the distribution of our centered test statistic $\sqrt{n}(\hat{\Delta}_{\tau}(S_1, S_2)-\Delta_{\tau}(S_1, S_2)).$ First, under the assumption that $S_1(t)$ equals $S_2(t)$ only on a set of Lebesgue measure zero, the standard bootstrap as well as a normal approximation (which corresponds to the oracle \eqref{oracletest}) are valid, and the arguments in Chapter 23.3 of \citet{vanderVaart1998} suggest, that the Type-I error of the corresponding tests is equal to $\alpha$ up to an error term of order $O_p(n^{-1/2})$ and $O(n^{-1/2})$ respectively, whereas this error term is at most $O_p(n^{-1/3})$ for the test based on subsampling \citep{politis1994large}. For the modified delta method it is possible to choose $c_n=1/\varepsilon_n=n^{1/2}$ in this case, which is asymptotically equivalent to the usual bootstrap (see Remark 3.5 of \citet{fang2019inference} and Section 5.2 of \citet{hong2018numerical}) and therefore the error is $O_p(n^{-1/2})$ as well.

If $S_1$ and $S_2$ coincide on an interval and therefore the test statistic is not asymptotically normal, the situation is less clear. However, we expect that the error terms depend on a combination of the underlying distribution and the chosen tuning parameters $b_j, c_n, \varepsilon_n$ and even an optimal choice might not recover the $O_p(n^{-1/2})$ rate, which is expected to be the best achievable, as we use a statistic which is not asymptotically pivotal, see Section 10 of \citet{politis1999subsampling} and the references therein for a discussion in the one-sample case. As the optimal choice requires knowledge of the underlying distribution, we cannot guarantee any rate of convergence for the Type-I error to $\alpha$, just that the convergence to $\alpha$ holds.

Analyzing (ii) is easiest for the test based on the estimated derivative \eqref{estimated_derivative}, but we expect similar results to hold for the numerical derivative as well as the subsampling approach. It follows from the definition of the limiting random variable $D_{\Delta_{\tau}}$, that it always dominates a normal distribution with the same variance in the usual stochastic order, due to the nonlinear component which is always nonnegative. As the corresponding term in the estimator is usually strictly positive for a fixed sample size, the bootstrap approximations will be too large on average if the true survival functions do not coincide on an interval, furthermore, the bootstrap approximation will have the wrong skewness making the estimation of the lower tail very challenging. Because this lower tail defines our rejection region, the bias in the Type-I error is positive in this case, and it increases the closer $S_1$ and $S_2$ are. To minimize this bias, $c_n$ should be chosen close to $n^{1/2}.$ However, the situation is reversed if $S_1(t)=S_2(t)$ on a large interval. In this case the estimator \eqref{estimated_derivative} might underestimate the limiting distribution, leading to a conservative finite sample behavior of the bootstrap test. This is most prominent if the survival functions are equal on the full interval $[0,\tau]$ in which case the linear part in $D_{\Delta_{\tau}}$ vanishes completely and its support is bounded from below. In this situation a very slowly growing choice for $c_n$ would improve the approximation, which suggests that  a good choice depends on the unknown survival functions and can not be made a priori without additional assumptions. 

In order to improve the finite sample properties of the proposed methods, we focus on possibilities to decrease the positive bias of our approximations, at the cost of possibly losing some power under the alternative. Even though it might be possible to choose the parameters $b_j, c_n, \varepsilon_n$ based on the data (see \citet{goetze2001adaptive} and \citet{bickel1997resampling} for some suggestions), we aim for data independent corrections, which have a stronger theoretical basis without adding additional assumptions. A method which is valid for all considered resampling approaches is to replace the used $\alpha$-quantiles with $\alpha_n$-quantiles where $\alpha_n=\alpha-C/n$ for an arbitrary constant $C>0.$ For the subsampling approach, we use an extrapolation procedure proposed by \citet{politis1999subsampling}, which we adopt to the two-sample case as follows: The subsampling distribution $L_{n,b}(x)$ is replaced by
\begin{equation*}
    L_{n,b}^{\ast}(x)=\frac{ L_{n,r_1}((1-r_1/n)^{1/2}x)(r_2^{-1/2}-n^{-1/2})-L_{n,r_2}((1-r_2/n)^{1/2}x)(r_1^{-1/2}-n^{-1/2})}{(r_2^{-1/2}-n^{-1/2})-(r_1^{-1/2}-n^{-1/2})}
\end{equation*}
where $r_1=C_1n^{2/3}, r_2=C_2n^{2/3}$ and $0<C_2<C_1.$ The within group subsamples are drawn proportional to the original sample sizes, without making this explicit in the notation. This correction is expected to recover the $O_p(n^{-1/2})$ rate at least in the asymptotically normal case. However, as the convergence rate of the Type-I error is unknown in the non-normal case, it is difficult to judge whether this adjustment is beneficial in general. 

As already discussed, for the modified bootstrap based on the estimated derivative \eqref{estimated_derivative}, a choice of $c_n$ close to $n^{1/2}$ leads to a decrease in positive bias under asymptotic normality, at the cost of being possibly conservative if asymptotic normality fails. We therefore propose to use $c_n=n^{1/(2+s)}$ for some $s\in(0,1).$ For the numerical delta method, a direct interpretable relation between $\varepsilon_n$ and the bias of the bootstrap test is less clear.
However, we expect a similar behavior as for the choice of $c_n$ in the sense that values closer to $n^{-1/2}$ will lead to a more conservative behavior. Furthermore, \citet{hong2018numerical} suggest a 2-point numerical derivative based on a higher order directional Taylor expansion for bias reduction, i.e., here
    \begin{equation}\label{numdiff_biasadjust}
    \begin{split}
        \Tilde{\Psi}_2(S_1,S_2)=\Big(&-0.5\Psi\big((\hat{S}_1-\hat{S}_2)+2\varepsilon_n\sqrt{n}((\hat{S}_1^{(B)}-\hat{S}_2^{(B)})\big)\\ &+2\Psi\big((\hat{S}_1-\hat{S}_2)+\varepsilon_n\sqrt{n}((\hat{S}_1^{(B)}-\hat{S}_2^{(B)}) -(\hat{S}_1-\hat{S}_2))\big)\\&-1.5\Psi(\hat{S}_1-\hat{S}_2)\Big)\varepsilon_n^{-1}.
    \end{split}
    \end{equation}
\section{Simulations}
\subsection{Simulation Setup}
To assess the finite sample properties of our proposed methods, especially type-I errors and power of the equivalence test for the hypothesis \eqref{equivalencehypothesis}, we conducted a simulation study. The following tests were included:
\begin{enumerate}[label=(\arabic*)]
    \item The subsampling test $\varphi^S_{n,b,\alpha}$ based on the extrapolation distribution $L^{\ast}_{n,b}$ with $C_1=2, C_2=1.$
    \item The corrected bootstrap test $\varphi^B_{n,\alpha}$ with the choice $c_n=n^{2.1}.$
    \item Bootstrap tests based numerical derivatives $\eqref{num_diff}$ and $\eqref{numdiff_biasadjust}$ which we denote by $\varphi^{N(1)}_{n,\alpha}$ and $\varphi^{N(2)}_{n,\alpha}$, respectively, $\varepsilon_n$ was chosen as $1/c_n=1/n^{1/2.1}$ for both tests.
    \item A naive bootstrap test $\varphi^E_{n,\alpha}$ based on conditional quantiles of \eqref{naive_boot}.
\end{enumerate}
For all tests, we used the more conservative $\alpha_n$-quantiles described in Section~\ref{sec:corrections_typeI_errors} with $C=1.$ It is worth noting that $\varphi^E_{n,\alpha}$ is only asymptotically exact if $S_1$ and $S_2$ do not coincide on an interval. However, it follows from the stochastic dominance of $D_{\Delta_{\tau}}$ over a normal distribution that the asymptotic size of this test is bounded by $\alpha$. 
In addition to type-I error rates and power, we also investigate coverage probabilities for two-sided confidence intervals for the effect measure $\Delta_{\tau}(S_1,S_2)$. To confine confidence intervals to the interval $(0,1),$ the complementary log-log (cloglog) transformation was used. 
As a consequence, this interval is only valid when $0<\Delta_{\tau}(S_1,S_2)<1.$ We denote these confidence intervals by $\text{CI}^S_{n,b,\alpha}$ etc.

Furthermore, all considered test were also combined with this transformation, based on one-sided intervals. These tests are denoted with an additional 'cl' in the superscript. In the display of the results, we mostly suppress the subscripts $n,\alpha,b$ to improve readability.

All simulations were carried out utilizing the R-computing environment, version 4.5.0 \citep{Rsoftware}, with $N_{\text{sim}}=5000$ simulation runs. Moreover for the resampling tests, $B=1000$ subsampling or bootstrap runs in each simulation step were used to approximate the respective distributions. Such a Monte-Carlo approach can also be justified theoretically, see Proposition 4.1 of \citet{romano1989bootstrap} and Lemma 5.2 of \citet{bucher2019note}.

For data generation, three different settings were considered; Weibull distributions with crossing/propor\-tio\-nal survival functions and piecewise exponential distributions with a certain overlap. Note that the two Weibull settings lead to asymptotically normally distributed test statistics, whereas the piecewise exponential settings results in a true mixture distribution. We chose the equivalence threshold $\varepsilon=0.1$. That is, all parameters have been chosen so that $\Delta_{\tau}(S_1,S_2)=0.1$ under $H_0.$ The resulting survival functions are displayed in Figure~\ref{fig:distH0}.
\begin{figure}[H]
\centering
\includegraphics[width=.3\textwidth]{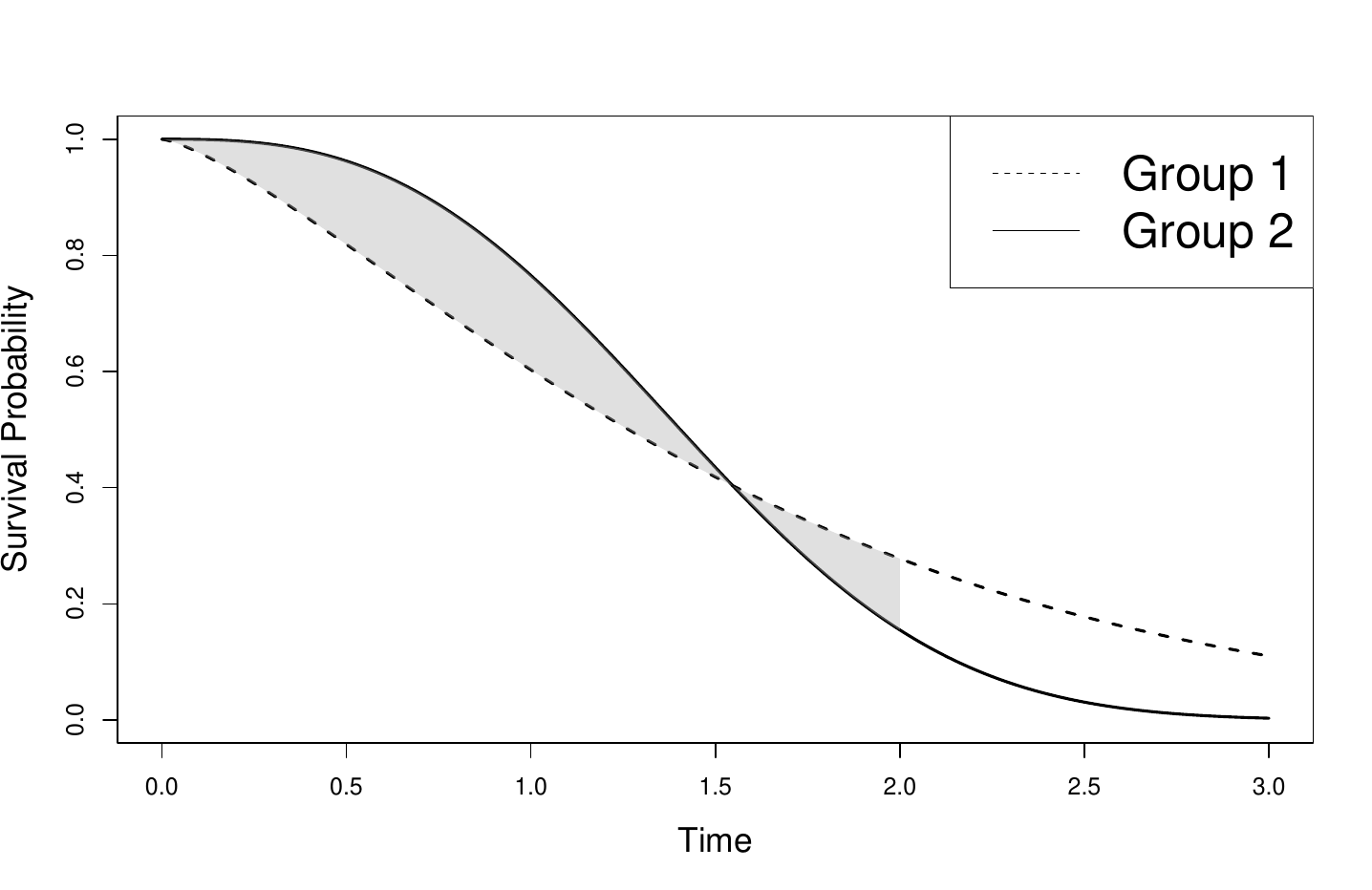}\hfill
\includegraphics[width=.3\textwidth]{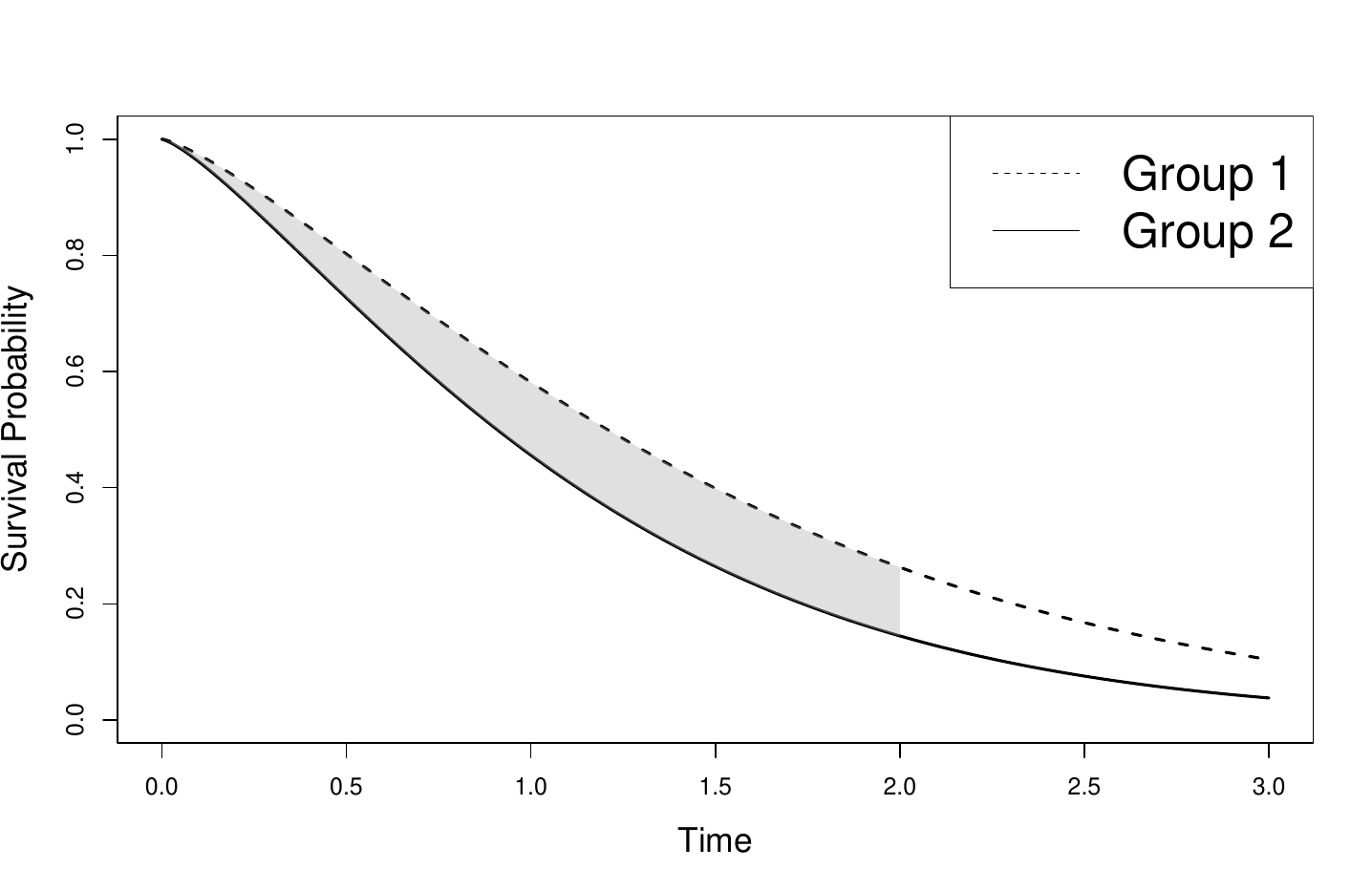}\hfill
\includegraphics[width=.3\textwidth]{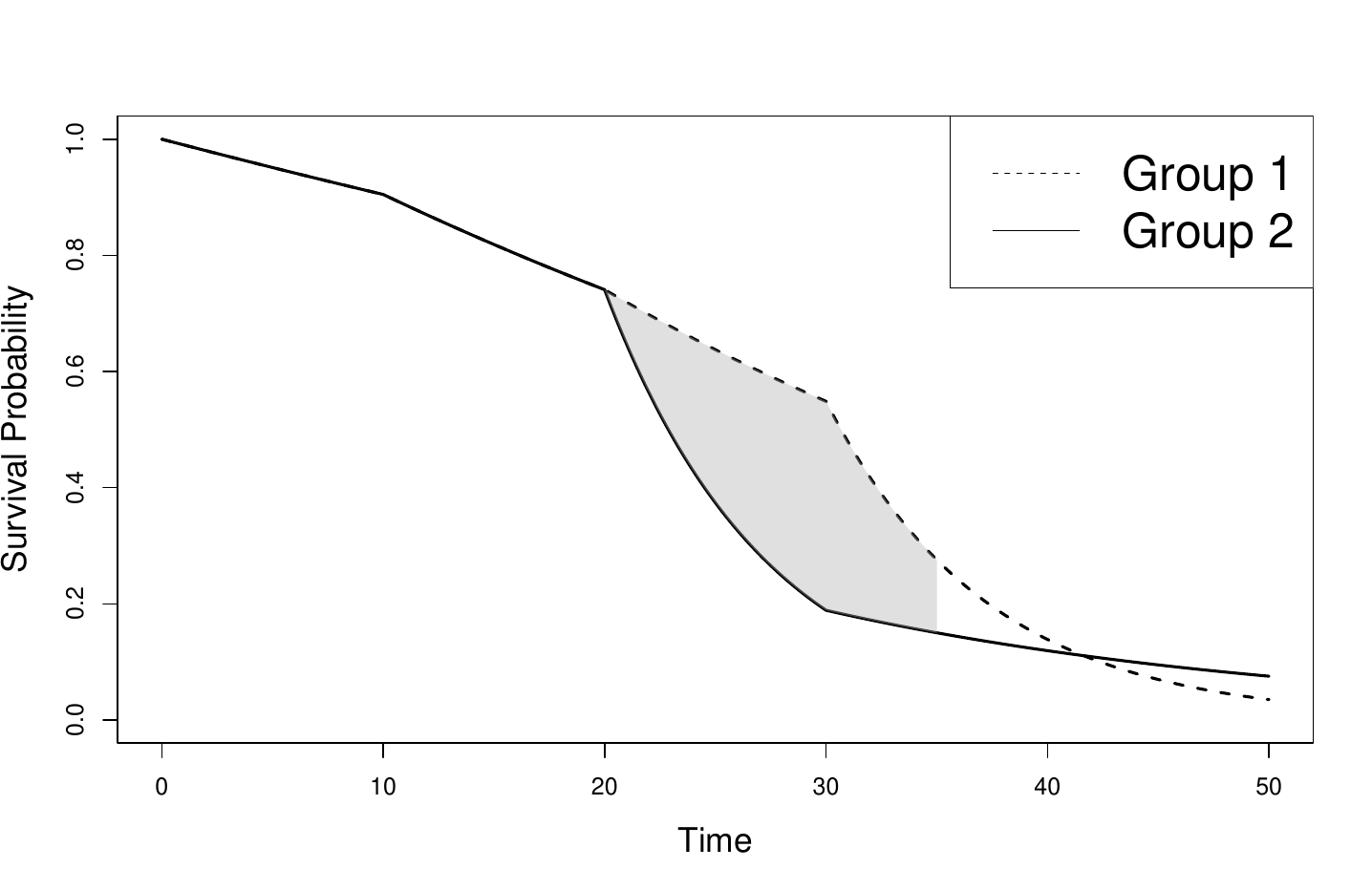}

\caption{Crossing and proportional Weibull as well as piecewise exponential distributions under the null hypothesis. The total area up to the terminal time $\tau$ is shaded.}
\label{fig:distH0}
\end{figure}
For the alternative, the parameters were chosen such that $\Delta_{\tau}(S_1,S_2)=0.05$; the resulting survival functions are displayed in Figure \ref{fig:distH1}. The terminal time $\tau=2$ was chosen under the Weibull distributions, and $\tau=35$ under the piecewise exponential distributions.

\begin{figure}[H]
\centering
\includegraphics[width=.3\textwidth]{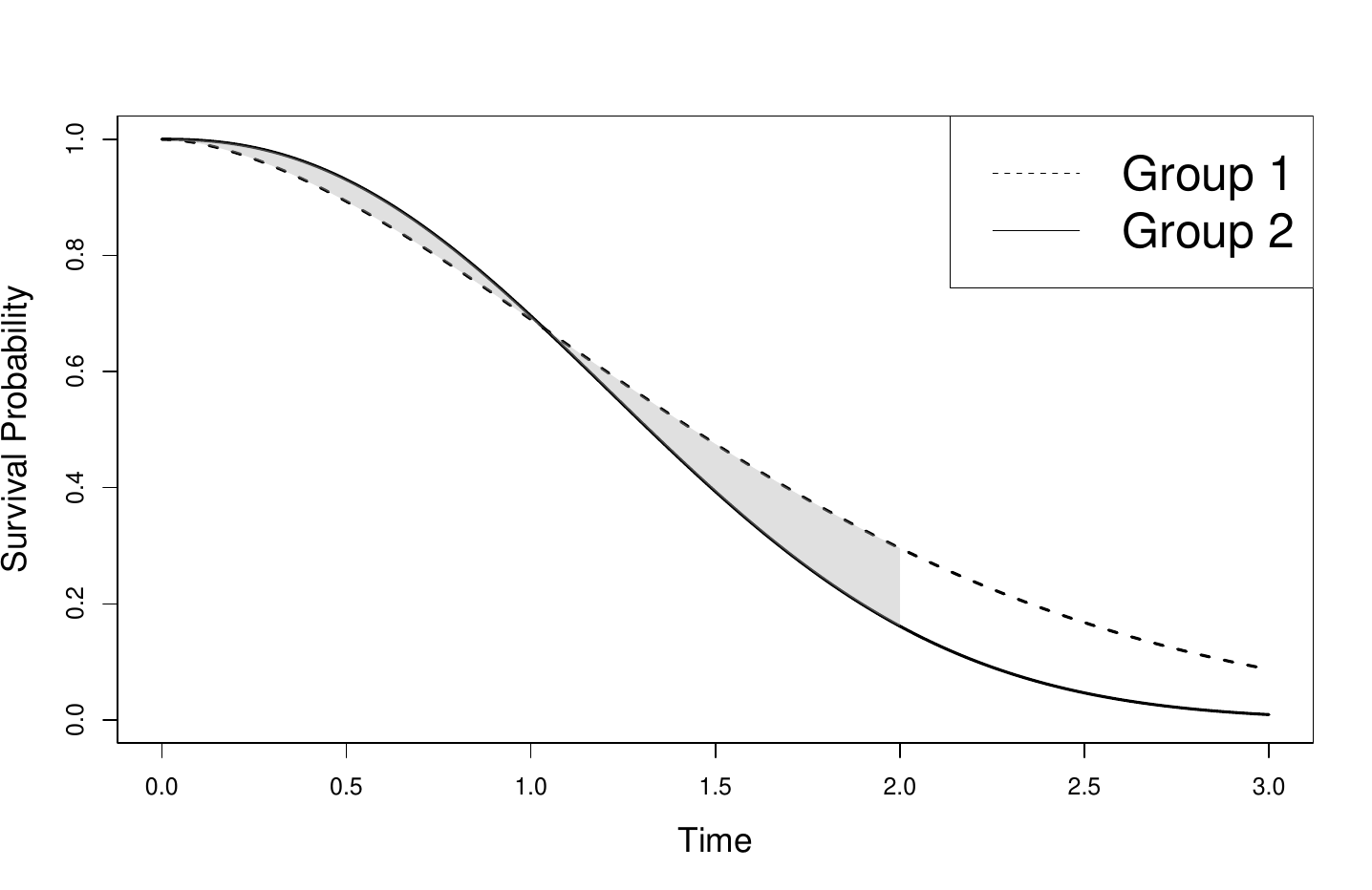}\hfill
\includegraphics[width=.3\textwidth]{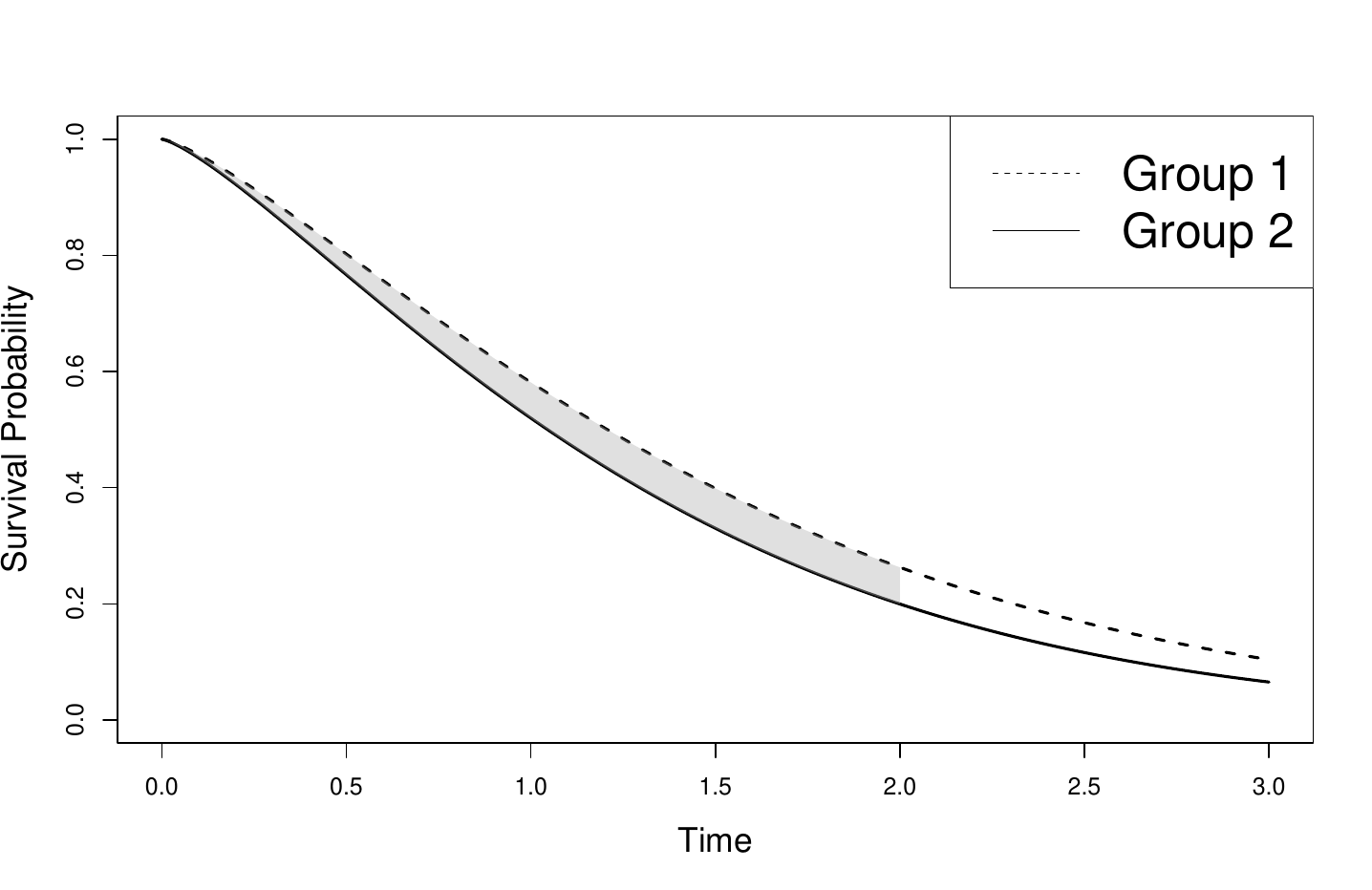}\hfill
\includegraphics[width=.3\textwidth]{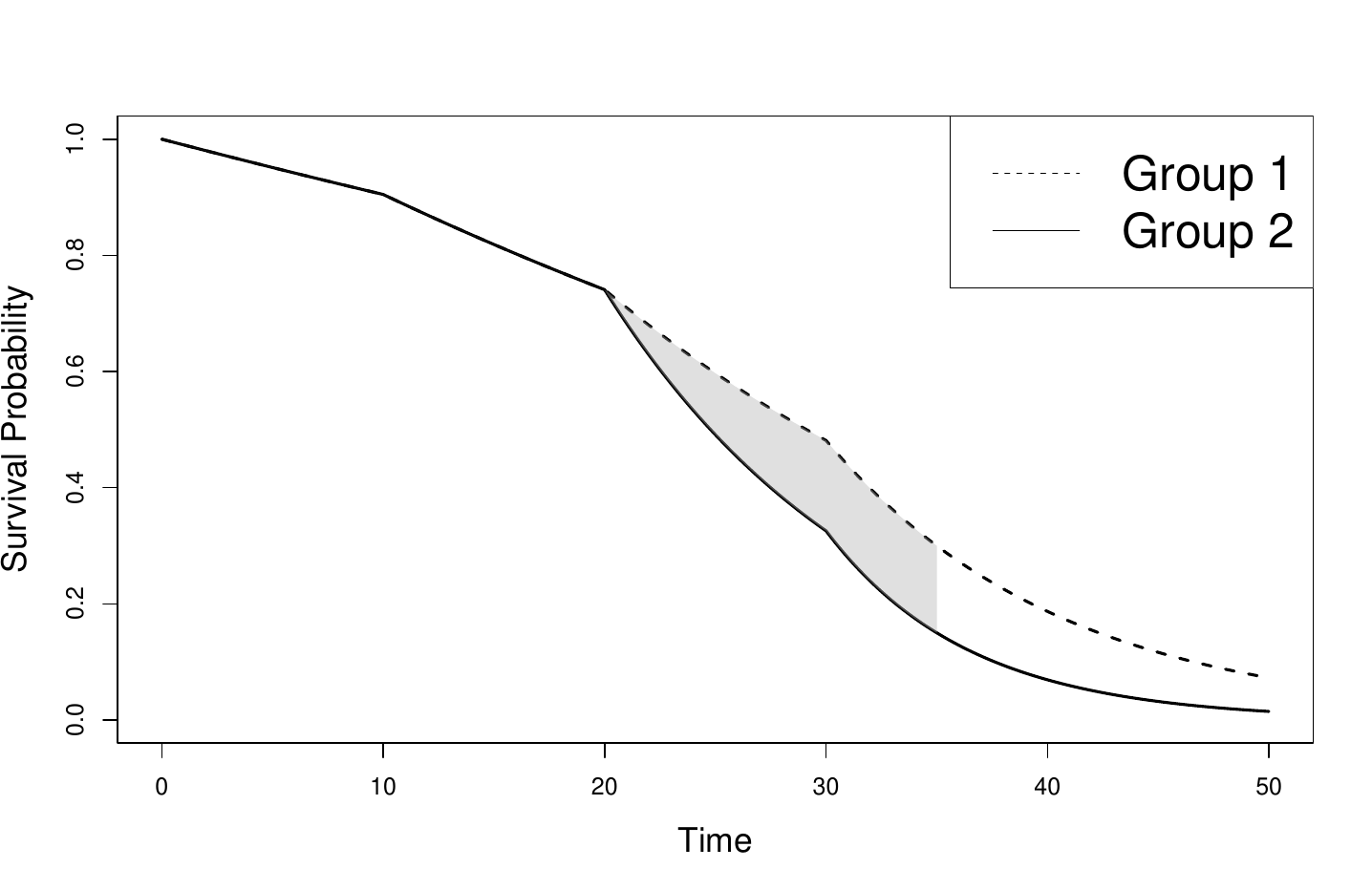}

\caption{Crossing and proportional Weibull as well as piecewise exponential distributions under the alternative hypothesis. The total area up to the terminal time $\tau$ is shaded.}
\label{fig:distH1}
\end{figure}
The exact parameters of the distributions have been obtained through nonlinear optimization with the COBYLA algorithm \citep{COBYLA, powell1998direct}, implemented in the NLopt package \citep{NLopt}. 
Each setting was simulated with and without right-censoring, where censoring times were simulated as independent exponential random variables, with different rate parameters chosen such that censoring probabilities for the two samples of $(\text{cens}_1, \text{cens}_2)= (0\%, 0\%), (30\%, 30\%), (60\%,60\%), \allowbreak (30\%,60\%)$ were achieved. We considered the sample sizes $(n_1, n_2)=(100,100), (100,200), (200,200),\allowbreak (250,500), (500,500)$.

\subsection{Results under the null hypothesis}
Table \ref{weib_crossing_h0} displays the results under $H_0$ based on Weibull distributions with crossing survival functions. Without the cloglog transformation, nearly all tests are too liberal in smaller samples; the  test $\varphi^{N(2)}$ based on the 2-point numerical derivative is an exception, it is too conservative in all settings. Even though the empirical size decreases for all methods in larger samples, only the subsampling test reaches the nominal level. However, it becomes conservative in settings with no or low censoring. The cloglog-transformation generally decreases the size of the tests, and the tests $\varphi^B$ and $\varphi^E$ keep the nominal level at least in larger samples, both in settings with and without censoring.
Regarding the coverage probabilities of the two-sided confidence intervals, the results are more satisfactory, especially the interval $\text{CI}^{N(2)}$ has close to nominal coverage across all sample sizes and censoring proportions.

\begin{table}[H]
\centering
\begin{adjustbox}{width=1\textwidth}
\small
\begin{tabular}{|rr|rr|*{5}{r}|*{5}{r}|*{5}{r}|}
\hline
$n_1$ & $n_2$ & $\text{cens}_1$ & $\text{cens}_2$ & $\varphi^S$ & $\varphi^B$ & $\varphi^{N(1)}$ & $\varphi^{N(2)}$ & $\varphi^E$ & $\varphi^S_{\text{cl}}$ & $\varphi^B_{\text{cl}}$ & $\varphi^{N(1)}_{\text{cl}}$ & $\varphi^{N(2)}_{\text{cl}}$ & $\varphi^{E}_{\text{cl}}$ & $\text{CI}^S$ & $\text{CI}^B$ & $\text{CI}^{N(1)}$ & $\text{CI}^{N(2)}$ & $\text{CI}^E$ \\
\hline
100 & 100 & 0\% & 0\% & 11.8 & 20.9 & 10.1 &  2.6 &  8.9 &  6.7 & 16.7 &  8.3 &  1.2 &  7.0 & 92.2 & 80.7 & 90.7 & 94.9 & 91.9 \\
100 & 200 & 0\% & 0\% & 10.5 & 16.1 &  9.8 &  2.4 &  8.6 & \textbf{5.0} & 11.1 &  7.7 &  1.0 &  6.3 & 94.5 & 86.7 & 92.8 & 96.0 & 93.4 \\
200 & 200 & 0\% & 0\% &  6.2 & 13.2 &  7.9 &  2.2 &  7.2 &  2.5 &  7.8 &  6.2 &  0.8 & \textbf{5.2} & 97.4 & 90.0 & 93.9 & 96.1 & 94.2 \\
250 & 500 & 0\% & 0\% &  4.1 & 10.6 &  7.4 &  3.0 &  7.0 &  1.1 & \textbf{5.6} &  6.1 &  1.7 & \textbf{5.5} & 98.9 & 92.5 & 94.3 & 96.5 & 94.6 \\
500 & 500 & 0\% & 0\% &  1.8 &  9.6 &  7.5 &  3.6 &  6.9 &  0.2 & \textbf{4.8} &  5.9 &  2.4 & \textbf{5.4} & 99.8 & 93.2 & 94.7 & 96.1 & 94.8 \\
\hline
100 & 100 & 30\% & 30\% & 13.8 & 18.3 & 10.5 &  3.0 &  9.4 &  8.8 & 14.4 &  9.0 &  1.6 &  7.5 & 90.0 & 82.0 & 89.6 & 93.9 & 90.5 \\
100 & 200 & 30\% & 30\% & 12.6 & 15.7 & 10.4 &  2.7 &  9.2 &  7.2 & 10.6 &  8.4 &  1.3 &  7.0 & 92.4 & 86.4 & 91.0 & 94.9 & 91.8 \\
200 & 200 & 30\% & 30\% &  8.3 & 12.9 &  8.1 &  2.3 &  7.3 &  3.7 &  7.8 &  6.4 &  0.9 & \textbf{5.6} & 96.2 & 89.8 & 90.7 & 94.7 & 91.4 \\
250 & 500 & 30\% & 30\% & \textbf{4.9} & 10.6 &  7.8 &  2.5 &  7.2 &  1.2 & \textbf{5.1} &  6.1 &  1.1 & \textbf{5.4} & 98.8 & 92.7 & 94.3 & 96.2 & 94.8 \\
500 & 500 & 30\% & 30\% &  2.6 &  9.3 &  7.0 &  2.9 &  6.5 &  0.5 & \textbf{4.5} &  5.7 &  1.5 & \textbf{5.1} & 99.5 & 93.3 & 94.7 & 96.1 & 94.9 \\
\hline
100 & 100 & 60\% & 60\% & 16.9 & 15.9 & 11.0 &  3.6 &  9.6 & 12.6 & 12.9 &  9.9 &  2.4 &  8.2 & 86.1 & 81.7 & 87.1 & 91.4 & 87.9 \\
100 & 200 & 60\% & 60\% & 16.6 & 13.6 & 11.1 &  3.4 &  9.5 & 11.5 &  9.9 &  9.6 &  2.1 &  8.0 & 87.8 & 85.0 & 87.6 & 92.3 & 88.6 \\
200 & 200 & 60\% & 60\% & 13.2 & 11.6 &  9.3 &  2.9 &  8.3 &  8.1 &  7.5 &  7.9 &  1.7 &  6.8 & 91.8 & 88.8 & 90.7 & 94.7 & 91.4 \\
250 & 500 & 60\% & 60\% &  9.4 &  8.9 &  8.3 &  2.1 &  7.3 & \textbf{4.5} & \textbf{4.5} &  6.9 &  1.2 & \textbf{5.6} & 95.5 & 92.1 & 92.8 & 95.4 & 93.5 \\
500 & 500 & 60\% & 60\% &  6.5 &  8.8 &  7.8 &  2.4 &  6.8 &  2.3 &  4.1 &  6.2 &  1.4 & \textbf{5.6} & 97.7 & 92.4 & 93.1 & 95.1 & 93.7 \\
\hline
100 & 100 & 30\% & 60\% & 15.1 & 17.1 & 10.3 &  3.3 &  8.8 & 10.2 & 12.9 &  8.9 &  2.0 &  7.4 & 88.4 & 83.1 & 89.1 & 93.7 & 89.9 \\
100 & 200 & 30\% & 60\% & 12.6 & 13.7 &  8.8 &  2.6 &  7.9 &  7.8 &  9.0 &  7.6 &  1.3 &  6.3 & 91.6 & 87.9 & 91.5 & 95.1 & 92.6 \\
200 & 200 & 30\% & 60\% & 11.0 & 12.8 &  9.0 &  2.3 &  7.8 &  5.9 &  7.7 &  7.4 &  1.2 &  6.4 & 93.9 & 88.6 & 91.6 & 94.7 & 92.2 \\
250 & 500 & 30\% & 60\% &  6.4 &  9.5 &  7.5 &  2.4 &  7.0 &  2.5 & \textbf{4.8} &  6.2 &  1.1 & \textbf{5.4} & 97.5 & 92.4 & 93.6 & 95.5 & 94.0 \\
500 & 500 & 30\% & 60\% & \textbf{4.5} &  8.9 &  7.2 &  2.5 &  6.7 &  1.2 & \textbf{4.5} &  6.0 &  1.5 & \textbf{5.4} & 98.8 & 92.8 & 94.1 & 95.9 & 94.6 \\
\hline
\end{tabular}
\end{adjustbox}
\caption{Empirical size and coverage probabilities in percent under $H_0$ based on Weibull distributions with crossing survival functions. Values for the size in the 95\% binomial interval $[0.044,0.056]$ are printed in bold.}
\label{weib_crossing_h0}
\end{table}

The setting based on Weibull distributions with proportional survival functions under $H_0$ appears to be more challenging for all considered methods; see the results in Table~\ref{weib_prop_h0}. No method controls the type-I error below sample sizes of $(200,200)$ with some methods exceeding 30\% in small samples. With increasing censoring this behavior worsens, such that acceptable but partially conservative results are only achieved for the largest sample size $(500,500)$ by the tests $\varphi^B_{\text{cl}}, \varphi^{E}_{\text{cl}}, \varphi^{N(2)}$ and $\varphi_{\text{cl}}^{N(2)}$. However, the general behavior of the tests is similar: the methods based on the 2-point numerical derivative are the most conservative and the cloglog transformation decreases the type-I error. The simulated empirical coverage probabilities of the confidence intervals also decrease, but at least in larger samples the two-point derivative and the subsampling approach lead to acceptable results.

\begin{table}[H]
\centering
\begin{adjustbox}{width=1\textwidth}
\small
\begin{tabular}{|rr|rr|*{5}{r}|*{5}{r}|*{5}{r}|}
\hline
$n_1$ & $n_2$ & $\text{cens}_1$ & $\text{cens}_2$ & $\varphi^S$ & $\varphi^B$ & $\varphi^{N(1)}$ & $\varphi^{N(2)}$ & $\varphi^E$ & $\varphi^S_{\text{cl}}$ & $\varphi^B_{\text{cl}}$ & $\varphi^{N(1)}_{\text{cl}}$ & $\varphi^{N(2)}_{\text{cl}}$ & $\varphi^{E}_{\text{cl}}$ & $\text{CI}^S$ & $\text{CI}^B$ & $\text{CI}^{N(1)}$ & $\text{CI}^{N(2)}$ & $\text{CI}^E$ \\
\hline
100 & 100 & 0\% & 0\% & 24.5 & 32.0 & 23.4 & 12.9 & 21.5 & 20.2 & 29.4 & 21.0 &  9.5 & 18.4 & 77.0 & 64.8 & 74.8 & 85.7 & 77.5 \\
100 & 200 & 0\% & 0\% & 20.5 & 23.7 & 19.9 &  9.7 & 17.8 & 15.9 & 19.9 & 17.0 &  6.4 & 14.7 & 82.3 & 74.4 & 78.8 & 88.8 & 81.1 \\
200 & 200 & 0\% & 0\% & 13.5 & 16.3 & 14.1 & \textbf{5.5} & 12.1 &  9.0 & 11.5 & 11.0 &  2.8 &  8.9 & 90.2 & 83.6 & 85.4 & 92.7 & 87.1 \\
250 & 500 & 0\% & 0\% &  6.8 &  9.4 &  8.9 &  2.3 &  7.6 &  3.6 &  4.2 &  5.7 &  0.6 & \textbf{4.7} & 96.0 & 91.2 & 90.8 & 94.9 & 92.0 \\
500 & 500 & 0\% & 0\% &  3.2 &  7.7 &  6.7 &  1.4 &  6.1 &  1.0 &  2.6 &  3.8 &  0.3 &  3.1 & 98.8 & 92.9 & 93.0 & 95.3 & 93.8 \\
\hline
100 & 100 & 30\% & 30\% & 28.8 & 31.2 & 25.8 & 14.9 & 24.2 & 24.8 & 28.6 & 24.1 & 11.6 & 21.8 & 72.2 & 64.8 & 70.9 & 83.2 & 73.6 \\
100 & 200 & 30\% & 30\% & 24.8 & 24.9 & 23.4 & 12.8 & 21.0 & 20.4 & 21.5 & 20.7 &  9.0 & 18.1 & 77.5 & 72.5 & 74.8 & 86.0 & 77.2 \\
200 & 200 & 30\% & 30\% & 17.5 & 18.1 & 16.9 &  7.1 & 14.8 & 12.8 & 13.0 & 13.7 &  4.2 & 10.8 & 86.3 & 81.5 & 82.4 & 90.6 & 85.0 \\
250 & 500 & 30\% & 30\% & 10.0 &  9.9 & 11.0 &  3.4 &  9.2 &  5.8 & \textbf{5.3} &  7.9 &  1.2 &  6.2 & 93.8 & 90.1 & 88.9 & 94.4 & 90.4 \\
500 & 500 & 30\% & 30\% & \textbf{5.2} &  7.7 &  7.7 &  1.5 &  6.8 &  2.1 &  2.4 & \textbf{4.5} &  0.5 &  3.5 & 97.6 & 93.0 & 92.1 & 95.1 & 93.4 \\
\hline
100 & 100 & 60\% & 60\% & 35.5 & 31.7 & 29.5 & 17.9 & 27.6 & 31.6 & 28.8 & 28.4 & 14.5 & 26.3 & 64.8 & 63.8 & 66.2 & 80.2 & 68.7 \\
100 & 200 & 60\% & 60\% & 33.9 & 27.3 & 28.4 & 17.0 & 26.1 & 29.6 & 23.7 & 26.7 & 13.3 & 24.1 & 67.8 & 69.1 & 67.9 & 81.2 & 70.6 \\
200 & 200 & 60\% & 60\% & 28.1 & 22.3 & 24.4 & 13.2 & 21.8 & 24.2 & 17.9 & 21.9 &  9.8 & 19.3 & 74.4 & 75.3 & 73.2 & 85.0 & 75.8 \\
250 & 500 & 60\% & 60\% & 21.2 & 14.4 & 18.5 &  8.0 & 16.3 & 16.3 &  9.0 & 15.6 & \textbf{5.2} & 12.1 & 83.2 & 85.6 & 81.2 & 90.2 & 84.1 \\
500 & 500 & 60\% & 60\% & 13.5 &  9.1 & 12.6 &  4.2 & 10.7 &  8.9 & \textbf{4.5} &  9.2 &  1.9 &  7.2 & 90.9 & 90.6 & 87.3 & 93.5 & 89.3 \\
\hline
100 & 100 & 30\% & 60\% & 34.8 & 32.9 & 30.6 & 18.8 & 28.3 & 31.3 & 29.6 & 28.8 & 15.0 & 26.2 & 64.2 & 62.6 & 65.2 & 79.0 & 67.6 \\
100 & 200 & 30\% & 60\% & 30.4 & 27.2 & 26.7 & 15.4 & 24.7 & 26.3 & 23.6 & 24.7 & 11.7 & 21.9 & 71.4 & 70.0 & 70.7 & 83.3 & 73.2 \\
200 & 200 & 30\% & 60\% & 24.3 & 20.4 & 21.2 & 10.8 & 18.9 & 19.9 & 15.8 & 18.6 &  7.4 & 15.8 & 78.8 & 78.0 & 77.0 & 87.2 & 79.5 \\
250 & 500 & 30\% & 60\% & 15.1 & 11.7 & 14.2 & \textbf{5.4} & 12.0 & 10.5 &  6.7 & 11.0 &  2.5 &  8.5 & 89.1 & 87.9 & 85.3 & 92.7 & 87.1 \\
500 & 500 & 30\% & 60\% & 10.4 &  8.5 & 10.8 &  3.4 &  9.0 &  6.3 &  3.6 &  7.7 &  1.4 & \textbf{5.5} & 93.4 & 91.2 & 88.9 & 93.7 & 90.3 \\
\hline
\end{tabular}
\end{adjustbox}
\caption{Empirical size and coverage probabilities in percent under $H_0$ based on Weibull distributions with proportional survival functions. Values for the size in the 95\% binomial interval $[0.044,0.056]$ are printed in bold.}
\label{weib_prop_h0}
\end{table}

Table \ref{pwexp_h0} displays the results based on piecewise exponential distributions under $H_0$. Here, nearly all tests without the cloglog transformation are in general too conservative, especially in larger samples. One exception is the test $\varphi^B$: all empirical sizes lie within the 95\% binomial interval for the largest considered sample sizes. Like with the Weibull settings, the cloglog transformation leads to an even more conservative behavior. As a consequence for all tests based on this transformation, the sizes are well below the nominal level for large samples. The empirical coverage probabilities of the confidence intervals are two-fold: the subsampling approach leads to coverages close to one for large samples, whereas the other methods do not reach the nominal coverage. Their behavior is very similar, and at least for the largest sample size coverages of 91-92\% are reached.

\begin{table}[H]
\centering
\begin{adjustbox}{width=1\textwidth}
\small
\begin{tabular}{|rr|rr|*{5}{r}|*{5}{r}|*{5}{r}|}
\hline
$n_1$ & $n_2$ & $\text{cens}_1$ & $\text{cens}_2$ & $\varphi^S$ & $\varphi^B$ & $\varphi^{N(1)}$ & $\varphi^{N(2)}$ & $\varphi^E$ & $\varphi^S_{\text{cl}}$ & $\varphi^B_{\text{cl}}$ & $\varphi^{N(1)}_{\text{cl}}$ & $\varphi^{N(2)}_{\text{cl}}$ & $\varphi^{E}_{\text{cl}}$ & $\text{CI}^S$ & $\text{CI}^B$ & $\text{CI}^{N(1)}$ & $\text{CI}^{N(2)}$ & $\text{CI}^E$ \\
\hline
100 & 100 & 0\%  & 0\%  & \textbf{5.4} &  9.2 &  4.1 &  0.9 &  3.7 &  3.0 &  6.2 &  3.6 &  0.4 &  2.9 & 93.2 & 86.7 & 88.6 & 89.3 & 88.9 \\
100 & 200 & 0\%  & 0\%  & \textbf{4.8} &  7.1 &  3.6 &  0.8 &  3.1 &  1.9 &  4.1 &  2.9 &  0.4 &  2.4 & 95.8 & 88.6 & 89.3 & 89.8 & 89.3 \\
200 & 200 & 0\%  & 0\%  &  2.7 &  5.8 &  2.7 &  0.8 &  2.4 &  0.8 &  2.9 &  2.1 &  0.4 &  1.8 & 98.6 & 90.9 & 91.3 & 91.2 & 91.2 \\
250 & 500 & 0\%  & 0\%  &  2.0 & \textbf{5.5} &  3.0 &  1.1 &  2.6 &  0.4 &  2.6 &  2.3 &  0.7 &  1.9 & 99.3 & 91.7 & 92.1 & 91.9 & 91.9 \\
500 & 500 & 0\%  & 0\%  &  1.2 & \textbf{5.6} &  2.7 &  1.5 &  2.4 &  0.2 &  2.5 &  2.3 &  1.1 &  2.1 & 99.7 & 92.1 & 92.4 & 92.2 & 92.3 \\
\hline
100 & 100 & 30\% & 30\% &  6.9 &  8.6 & \textbf{4.9} &  1.1 &  4.2 &  4.1 &  6.1 &  4.1 &  0.5 &  3.3 & 92.6 & 86.6 & 88.0 & 89.5 & 88.3 \\
100 & 200 & 30\% & 30\% &  6.0 &  7.6 &  4.2 &  1.1 &  3.7 &  2.9 &  4.3 &  3.7 &  0.5 &  2.9 & 94.7 & 88.1 & 88.7 & 89.7 & 89.0 \\
200 & 200 & 30\% & 30\% &  3.8 &  6.4 &  3.3 &  0.7 &  2.8 &  1.2 &  3.1 &  2.6 &  0.2 &  2.0 & 97.9 & 90.2 & 90.6 & 90.9 & 90.5 \\
250 & 500 & 30\% & 30\% &  2.0 & \textbf{5.3} &  2.7 &  0.9 &  2.5 &  0.5 &  2.2 &  2.4 &  0.6 &  1.9 & 99.1 & 91.8 & 91.8 & 91.7 & 91.7 \\
500 & 500 & 30\% & 30\% &  1.4 & \textbf{5.2} &  2.8 &  1.4 &  2.6 &  0.2 &  2.4 &  2.2 &  0.9 &  1.8 & 99.7 & 91.7 & 91.8 & 91.8 & 91.7 \\
\hline
100 & 100 & 60\% & 60\% & 12.3 & 10.6 &  8.2 &  2.6 &  7.3 &  8.9 &  7.6 &  7.5 &  1.7 &  6.1 & 86.6 & 83.6 & 83.5 & 86.8 & 84.1 \\
100 & 200 & 60\% & 60\% & 10.5 &  8.9 &  7.0 &  1.7 &  5.9 &  6.9 & \textbf{5.2} &  5.9 &  1.1 & \textbf{5.0} & 90.6 & 86.7 & 85.7 & 88.1 & 86.1 \\
200 & 200 & 60\% & 60\% &  7.2 &  6.9 & \textbf{5.5} &  1.1 & \textbf{4.8} & \textbf{4.5} &  3.8 & \textbf{4.6} &  0.4 &  3.8 & 94.0 & 88.3 & 87.6 & 89.8 & 88.3 \\
250 & 500 & 60\% & 60\% & \textbf{4.6} &  5.8 & \textbf{4.5} &  1.0 &  3.7 &  1.7 &  2.5 &  3.4 &  0.4 &  2.9 & 97.8 & 90.3 & 89.8 & 90.9 & 90.0 \\
500 & 500 & 60\% & 60\% &  2.3 & \textbf{5.2} &  3.5 &  1.0 &  3.0 &  0.5 &  1.9 &  2.7 &  0.4 &  2.0 & 99.3 & 91.7 & 91.1 & 91.4 & 91.0 \\
\hline
100 & 100 & 30\% & 60\% &  9.7 &  9.8 &  6.9 &  1.8 &  6.0 &  6.7 &  6.3 &  6.0 &  0.9 & \textbf{4.8} & 89.2 & 84.5 & 84.2 & 87.1 & 85.1 \\
100 & 200 & 30\% & 60\% &  7.2 &  7.5 & \textbf{5.2} &  1.2 &  4.3 &  3.8 &  4.3 &  4.1 &  0.5 &  3.3 & 94.0 & 87.6 & 87.6 & 89.2 & 87.9 \\
200 & 200 & 30\% & 60\% &  5.8 &  6.9 & \textbf{4.7} &  0.9 &  4.1 &  2.9 &  3.7 &  3.7 &  0.4 &  2.9 & 96.2 & 89.0 & 88.7 & 90.1 & 89.0 \\
250 & 500 & 30\% & 60\% &  2.9 &  5.8 &  3.8 &  1.0 &  3.3 &  0.7 &  2.5 &  2.8 &  0.4 &  2.3 & 98.9 & 90.9 & 90.8 & 91.2 & 91.0 \\
500 & 500 & 30\% & 60\% &  1.4 & \textbf{4.8} &  2.6 &  0.9 &  2.3 &  0.2 &  1.7 &  2.0 &  0.4 &  1.7 & 99.6 & 92.1 & 91.0 & 91.1 & 91.1 \\
\hline
\end{tabular}
\end{adjustbox}
\caption{Empirical size and coverage probabilities in percent under $H_0$ based on piecewise exponential distributions. Values for the size in the 95\% binomial interval $[0.044,0.056]$ are printed in bold.}
\label{pwexp_h0}
\end{table}
In both settings based on Weibull distributions, an increase in censoring proportion generally increases the type-I error of tests without the cloglog transformation, however, the tests based on the transformation seem less affected. For the setting based on piecewise exponential distributions, the size of the tests does not vary much with the censoring rate.

\subsection{Results under the alternative}
Table \ref{weib_crossing_h1} displays the results under $H_1$ based on Weibull distributions with crossing survival functions. The test $\varphi^{N(2)}$ and all transformed tests besides $\varphi^{N(1)}_{\text{cl}}$ did not exceed the nominal level under $H_0,$ at least in larger samples, therefore we focus on those regarding power comparisons. The tests based on the cloglog transformation have in general a lower power compared to their untransformed counterparts, which is not surprising given the more conservative behavior under $H_0.$ The power of these tests a fairly similar, with the exception of $\varphi^S_{\text{cl}}$ which has noticeably lower power in larger samples, which is in accordance with the very low size of this test under $H_0.$
The empirical coverage probabilities of all confidence intervals are noticeably worse compared to the same settings under $H_0.$ This can be attributed to the fact that $\Delta_{\tau}(S_1,S_2)=0.05,$ is closer to the boundary of 0, which makes the confidence intervals more unstable. Without censoring, the intervals $\text{CI}^S$ and $\text{CI}^E$ still provide acceptable coverage in large samples. However, under high censoring proportions, no confidence interval even reaches 90\% coverage.

\begin{table}[H]
\centering
\begin{adjustbox}{width=1\textwidth}
\small
\begin{tabular}{|rr|rr|*{5}{r}|*{5}{r}|*{5}{r}|}
\hline
$n_1$ & $n_2$ & $\text{cens}_1$ & $\text{cens}_2$ & $\varphi^S$ & $\varphi^B$ & $\varphi^{N(1)}$ & $\varphi^{N(2)}$ & $\varphi^E$ & $\varphi^S_{\text{cl}}$ & $\varphi^B_{\text{cl}}$ & $\varphi^{N(1)}_{\text{cl}}$ & $\varphi^{N(2)}_{\text{cl}}$ & $\varphi^{E}_{\text{cl}}$ & $\text{CI}^S$ & $\text{CI}^B$ & $\text{CI}^{N(1)}$ & $\text{CI}^{N(2)}$ & $\text{CI}^E$ \\
\hline
100 & 100 & 0\%  & 0\%  & 75.8 & 82.7 & 87.9 & 73.6 & 71.2 & 66.2 & 77.6 & 87.3 & 70.0 & 65.5 & 79.3 & 60.4 & 68.0 & 79.8 & 82.7 \\
100 & 200 & 0\%  & 0\%  & 80.9 & 86.0 & 92.5 & 80.0 & 78.1 & 71.3 & 79.4 & 92.2 & 76.3 & 73.3 & 79.9 & 67.6 & 68.0 & 80.0 & 84.6 \\
200 & 200 & 0\%  & 0\%  & 87.2 & 92.7 & 96.5 & 87.9 & 88.2 & 75.9 & 86.8 & 96.2 & 84.3 & 84.3 & 83.9 & 71.5 & 69.5 & 82.7 & 87.4 \\
250 & 500 & 0\%  & 0\%  & 94.3 & 98.0 & 99.3 & 95.5 & 96.6 & 84.6 & 95.1 & 99.1 & 93.6 & 95.1 & 88.3 & 79.8 & 71.1 & 85.9 & 90.0 \\
500 & 500 & 0\%  & 0\%  & 98.2 & 99.6 & 99.9 & 98.9 & 99.4 & 91.5 & 99.1 & 99.9 & 98.3 & 99.1 & 93.4 & 84.8 & 75.0 & 90.6 & 93.1 \\
\hline
100 & 100 & 30\% & 30\% & 71.8 & 75.1 & 84.2 & 69.4 & 65.3 & 64.2 & 68.4 & 83.9 & 66.4 & 60.1 & 75.3 & 63.6 & 66.8 & 76.9 & 80.8 \\
100 & 200 & 30\% & 30\% & 76.1 & 78.4 & 88.8 & 75.7 & 71.4 & 67.4 & 70.3 & 88.3 & 72.6 & 66.5 & 76.5 & 67.3 & 64.9 & 76.4 & 81.7 \\
200 & 200 & 30\% & 30\% & 83.9 & 87.7 & 94.3 & 84.1 & 82.5 & 72.9 & 78.9 & 93.9 & 80.8 & 76.9 & 79.2 & 71.3 & 66.1 & 78.7 & 85.1 \\
250 & 500 & 30\% & 30\% & 90.3 & 95.0 & 98.3 & 92.5 & 92.9 & 79.9 & 88.9 & 98.2 & 90.0 & 89.8 & 84.7 & 80.3 & 67.8 & 82.7 & 88.2 \\
500 & 500 & 30\% & 30\% & 95.2 & 98.9 & 99.6 & 97.3 & 98.1 & 85.5 & 96.7 & 99.6 & 95.6 & 97.2 & 89.0 & 83.0 & 70.8 & 86.2 & 90.6 \\
\hline
100 & 100 & 60\% & 60\% & 67.0 & 62.3 & 77.5 & 62.8 & 55.7 & 60.1 & 55.2 & 77.4 & 60.7 & 52.3 & 70.1 & 64.2 & 63.9 & 72.9 & 77.5 \\
100 & 200 & 60\% & 60\% & 70.9 & 64.9 & 82.5 & 69.2 & 61.7 & 63.6 & 56.5 & 82.3 & 67.2 & 57.7 & 70.2 & 67.3 & 61.8 & 71.7 & 77.2 \\
200 & 200 & 60\% & 60\% & 78.0 & 74.5 & 90.0 & 78.6 & 72.0 & 70.6 & 63.7 & 89.8 & 76.3 & 66.9 & 70.8 & 69.6 & 60.7 & 72.0 & 79.2 \\
250 & 500 & 60\% & 60\% & 84.1 & 85.1 & 95.3 & 86.7 & 82.9 & 75.5 & 71.9 & 95.2 & 84.7 & 77.7 & 73.9 & 77.3 & 60.5 & 73.2 & 82.8 \\
500 & 500 & 60\% & 60\% & 88.6 & 93.0 & 98.4 & 92.2 & 91.2 & 78.9 & 84.5 & 98.3 & 90.3 & 87.3 & 78.8 & 81.8 & 63.6 & 77.1 & 86.0 \\
\hline
100 & 100 & 30\% & 60\% & 67.9 & 66.3 & 80.0 & 64.9 & 58.5 & 60.9 & 59.0 & 79.8 & 62.4 & 54.0 & 71.7 & 63.8 & 64.1 & 73.8 & 77.5 \\
100 & 200 & 30\% & 60\% & 75.1 & 73.3 & 87.1 & 74.1 & 68.3 & 67.1 & 64.3 & 86.9 & 71.7 & 63.1 & 74.0 & 68.9 & 64.6 & 75.3 & 80.9 \\
200 & 200 & 30\% & 60\% & 79.0 & 78.8 & 91.9 & 79.4 & 74.5 & 69.9 & 68.2 & 91.5 & 76.5 & 68.9 & 75.0 & 71.3 & 63.3 & 74.9 & 81.6 \\
250 & 500 & 30\% & 60\% & 87.8 & 91.5 & 97.5 & 90.4 & 89.2 & 77.7 & 82.7 & 97.3 & 88.3 & 85.0 & 79.8 & 79.2 & 64.6 & 78.4 & 86.0 \\
500 & 500 & 30\% & 60\% & 91.9 & 96.3 & 99.0 & 94.7 & 94.9 & 81.4 & 91.0 & 98.9 & 92.8 & 92.7 & 83.8 & 82.9 & 66.3 & 80.9 & 88.3 \\
\hline
\end{tabular}
\end{adjustbox}
\caption{Empirical power and coverage probabilities in percent under $H_1$ based on Weibull distributions with crossing survival functions.}
\label{weib_crossing_h1}
\end{table}

 Similar to the behavior under $H_0,$ the performance of all considered methods under $H_1$ with proportional Weibull survival functions displayed in Table \ref{weib_prop_h1} is worse compared to the setup with crossing survival functions, i.e., lower empirical power and coverage probabilities. Furthermore, it is remarkable for $\varphi^S_{\text{cl}}$ and $\varphi^B_{\text{cl}}$ that their power does not increase with sample size; for some settings it even decreases. A possible explanation is the possible instability of the cloglog-transformed confidence intervals at the lower boundary, combined with the high empirical type-I error in small samples. This could also partly explain the low empirical coverage probabilities of about 70\% in high censoring regimes, even for the largest samples.

\begin{table}[H]
\centering
\begin{adjustbox}{width=1\textwidth}
\small
\begin{tabular}{|rr|rr|*{5}{r}|*{5}{r}|*{5}{r}|}
\hline
$n_1$ & $n_2$ & $\text{cens}_1$ & $\text{cens}_2$ & $\varphi^S$ & $\varphi^B$ & $\varphi^{N(1)}$ & $\varphi^{N(2)}$ & $\varphi^E$ & $\varphi^S_{\text{cl}}$ & $\varphi^B_{\text{cl}}$ & $\varphi^{N(1)}_{\text{cl}}$ & $\varphi^{N(2)}_{\text{cl}}$ & $\varphi^{E}_{\text{cl}}$ & $\text{CI}^S$ & $\text{CI}^B$ & $\text{CI}^{N(1)}$ & $\text{CI}^{N(2)}$ & $\text{CI}^E$ \\
\hline
100 & 100 & 0\%  & 0\%  & 65.0 & 73.0 & 78.1 & 67.2 & 61.2 & 59.3 & 70.3 & 77.7 & 64.4 & 56.8 & 60.7 & 40.7 & 49.2 & 59.4 & 66.3 \\
100 & 200 & 0\%  & 0\%  & 66.4 & 70.3 & 81.6 & 70.4 & 63.1 & 60.2 & 65.7 & 81.1 & 67.6 & 57.9 & 61.5 & 46.4 & 48.2 & 58.7 & 66.0 \\
200 & 200 & 0\%  & 0\%  & 68.0 & 70.9 & 85.7 & 74.2 & 65.1 & 58.7 & 63.5 & 85.1 & 70.7 & 58.3 & 66.2 & 52.1 & 50.6 & 60.8 & 69.1 \\
250 & 500 & 0\%  & 0\%  & 69.9 & 73.9 & 91.6 & 80.9 & 71.2 & 57.1 & 58.1 & 91.2 & 76.8 & 61.1 & 72.3 & 64.2 & 53.0 & 64.4 & 75.0 \\
500 & 500 & 0\%  & 0\%  & 72.2 & 83.3 & 95.3 & 85.3 & 80.5 & 54.5 & 62.0 & 94.8 & 80.7 & 68.7 & 79.2 & 71.8 & 56.4 & 69.1 & 80.9 \\
\hline
100 & 100 & 30\% & 30\% & 65.6 & 68.5 & 76.5 & 65.7 & 59.9 & 60.6 & 65.6 & 76.4 & 63.8 & 56.7 & 59.5 & 45.8 & 50.2 & 60.8 & 67.0 \\
100 & 200 & 30\% & 30\% & 67.3 & 67.3 & 79.9 & 69.7 & 62.7 & 61.9 & 62.6 & 79.6 & 67.7 & 58.5 & 59.0 & 48.2 & 47.4 & 57.1 & 65.3 \\
200 & 200 & 30\% & 30\% & 68.7 & 68.7 & 84.7 & 73.6 & 64.4 & 60.9 & 61.3 & 84.3 & 70.7 & 58.3 & 63.3 & 52.4 & 48.7 & 59.3 & 67.8 \\
250 & 500 & 30\% & 30\% & 70.8 & 70.9 & 90.8 & 80.2 & 69.5 & 59.1 & 55.6 & 90.3 & 76.7 & 59.5 & 68.0 & 63.0 & 51.1 & 61.3 & 72.4 \\
500 & 500 & 30\% & 30\% & 73.1 & 78.5 & 93.9 & 84.3 & 76.4 & 59.0 & 58.1 & 93.2 & 80.5 & 65.2 & 73.6 & 68.7 & 52.0 & 64.3 & 76.2 \\
\hline
100 & 100 & 60\% & 60\% & 62.9 & 58.0 & 71.3 & 59.7 & 53.1 & 58.4 & 54.0 & 71.2 & 58.7 & 50.6 & 63.1 & 57.7 & 57.2 & 67.2 & 72.9 \\
100 & 200 & 60\% & 60\% & 67.1 & 60.1 & 76.1 & 66.5 & 59.4 & 62.8 & 55.0 & 76.1 & 65.2 & 56.1 & 59.0 & 57.0 & 51.5 & 60.9 & 68.7 \\
200 & 200 & 60\% & 60\% & 69.2 & 63.0 & 80.0 & 70.3 & 62.5 & 64.9 & 56.7 & 79.7 & 69.0 & 58.3 & 58.1 & 55.9 & 47.4 & 57.7 & 67.3 \\
250 & 500 & 60\% & 60\% & 73.5 & 64.4 & 87.1 & 78.2 & 67.5 & 66.8 & 53.5 & 86.8 & 76.4 & 61.2 & 58.3 & 61.1 & 45.3 & 55.3 & 66.2 \\
500 & 500 & 60\% & 60\% & 74.5 & 68.1 & 91.7 & 82.2 & 70.3 & 66.1 & 51.5 & 91.5 & 79.7 & 62.0 & 62.7 & 66.2 & 48.1 & 58.4 & 71.0 \\
\hline
100 & 100 & 30\% & 60\% & 64.7 & 63.1 & 73.3 & 62.8 & 57.4 & 60.1 & 59.0 & 73.2 & 61.5 & 54.3 & 59.4 & 51.8 & 52.4 & 63.2 & 69.3 \\
100 & 200 & 30\% & 60\% & 68.4 & 65.4 & 79.2 & 69.3 & 62.2 & 63.8 & 60.2 & 79.0 & 67.7 & 58.2 & 57.9 & 52.0 & 48.8 & 59.5 & 67.1 \\
200 & 200 & 30\% & 60\% & 69.1 & 64.7 & 82.5 & 71.5 & 63.0 & 63.6 & 57.8 & 82.2 & 69.7 & 58.2 & 59.8 & 54.7 & 48.6 & 58.6 & 67.2 \\
250 & 500 & 30\% & 60\% & 72.1 & 67.2 & 89.0 & 79.2 & 67.6 & 63.6 & 54.9 & 88.7 & 76.2 & 60.6 & 62.5 & 61.3 & 47.3 & 57.7 & 69.2 \\
500 & 500 & 30\% & 60\% & 74.0 & 72.1 & 92.2 & 82.8 & 72.1 & 63.9 & 54.0 & 91.8 & 80.1 & 63.1 & 66.8 & 66.9 & 48.3 & 60.0 & 72.0 \\
\hline
\end{tabular}
\end{adjustbox}
\caption{Empirical power and coverage probabilities in percent under $H_1$ based on Weibull distributions with proportional survival functions.}
\label{weib_prop_h1}
\end{table}
Table \ref{pwexp_h1} displays the results based on piecewise exponential distributions under $H_1.$ Among the untransformed tests, $\varphi^{N(1)}$ and $\varphi^B$ have the highest empirical power, whereas $\varphi^S$ has the lowest empirical power in large samples; the tests $\varphi^{N(2)}$ and $\varphi^E$ lie in between and have comparable empirical power. For the tests based the cloglog-transformation, the relationships are similar; as expected, due to the more conservative behavior under $H_0,$ the empirical power is generally lower. The empirical coverage probability of $\text{CI}^S$ is satisfactory at least in large samples, the other intervals exhibit severe undercoverage, reaching at most 90\% without censoring, and even lower values in high censoring regimes.

\begin{table}[H]
\centering
\begin{adjustbox}{width=1\textwidth}
\small
\begin{tabular}{|rr|rr|*{5}{r}|*{5}{r}|*{5}{r}|}
\hline
$n_1$ & $n_2$ & $\text{cens}_1$ & $\text{cens}_2$ & $\varphi^S$ & $\varphi^B$ & $\varphi^{N(1)}$ & $\varphi^{N(2)}$ & $\varphi^E$ & $\varphi^S_{\text{cl}}$ & $\varphi^B_{\text{cl}}$ & $\varphi^{N(1)}_{\text{cl}}$ & $\varphi^{N(2)}_{\text{cl}}$ & $\varphi^{E}_{\text{cl}}$ & $\text{CI}^S$ & $\text{CI}^B$ & $\text{CI}^{N(1)}$ & $\text{CI}^{N(2)}$ & $\text{CI}^E$ \\
\hline
100 & 100 & 0\%  & 0\%  & 64.7 & 74.9 & 78.7 & 60.9 & 59.6 & 54.6 & 67.1 & 77.9 & 56.7 & 54.4 & 82.6 & 69.6 & 71.9 & 79.4 & 81.4 \\
100 & 200 & 0\%  & 0\%  & 69.9 & 78.6 & 83.1 & 67.4 & 67.5 & 60.0 & 71.5 & 82.4 & 63.3 & 62.4 & 84.8 & 75.2 & 71.9 & 80.7 & 83.2 \\
200 & 200 & 0\%  & 0\%  & 76.4 & 86.6 & 89.2 & 75.4 & 76.9 & 63.7 & 81.3 & 88.7 & 70.7 & 72.2 & 91.3 & 81.3 & 76.6 & 84.9 & 86.6 \\
250 & 500 & 0\%  & 0\%  & 84.7 & 93.3 & 95.3 & 86.2 & 87.9 & 73.3 & 89.9 & 94.9 & 82.8 & 84.5 & 94.0 & 85.7 & 79.0 & 87.5 & 88.0 \\
500 & 500 & 0\%  & 0\%  & 90.6 & 96.6 & 98.1 & 93.2 & 94.1 & 81.0 & 95.0 & 97.9 & 91.0 & 92.4 & 96.5 & 87.9 & 81.9 & 90.0 & 89.7 \\
\hline
100 & 100 & 30\% & 30\% & 62.3 & 66.1 & 76.0 & 58.3 & 54.6 & 53.8 & 57.5 & 75.4 & 55.1 & 49.7 & 78.6 & 68.7 & 68.5 & 76.6 & 79.5 \\
100 & 200 & 30\% & 30\% & 66.3 & 71.6 & 80.4 & 64.8 & 61.7 & 56.7 & 61.8 & 79.9 & 61.4 & 56.3 & 81.7 & 73.6 & 69.2 & 77.6 & 80.5 \\
200 & 200 & 30\% & 30\% & 72.2 & 80.5 & 86.7 & 71.3 & 70.9 & 60.6 & 72.4 & 86.1 & 67.5 & 65.7 & 86.9 & 79.5 & 72.8 & 80.7 & 83.8 \\
250 & 500 & 30\% & 30\% & 81.1 & 89.6 & 93.7 & 82.4 & 83.6 & 67.7 & 84.4 & 93.2 & 78.8 & 80.1 & 91.1 & 84.8 & 75.0 & 84.6 & 86.6 \\
500 & 500 & 30\% & 30\% & 87.4 & 94.9 & 97.3 & 90.3 & 91.3 & 75.5 & 92.0 & 97.0 & 87.4 & 88.7 & 94.9 & 87.2 & 78.6 & 87.7 & 88.7 \\
\hline
100 & 100 & 60\% & 60\% & 57.9 & 53.6 & 70.5 & 54.5 & 47.8 & 52.3 & 46.0 & 70.3 & 52.4 & 44.4 & 72.9 & 67.0 & 65.1 & 72.1 & 75.8 \\
100 & 200 & 60\% & 60\% & 63.1 & 58.5 & 76.5 & 61.3 & 53.4 & 55.0 & 47.3 & 76.3 & 58.9 & 48.7 & 74.7 & 70.3 & 63.5 & 71.6 & 76.1 \\
200 & 200 & 60\% & 60\% & 67.5 & 67.8 & 82.9 & 67.4 & 61.1 & 58.8 & 54.7 & 82.5 & 64.5 & 56.0 & 77.9 & 74.7 & 64.5 & 73.4 & 78.9 \\
250 & 500 & 60\% & 60\% & 74.1 & 79.5 & 90.2 & 76.8 & 73.0 & 62.1 & 68.2 & 89.8 & 73.9 & 67.6 & 83.3 & 81.4 & 67.1 & 76.9 & 82.4 \\
500 & 500 & 60\% & 60\% & 79.9 & 88.4 & 94.6 & 83.8 & 83.0 & 66.0 & 80.9 & 94.2 & 80.1 & 78.6 & 88.1 & 85.4 & 70.4 & 80.9 & 85.8 \\
\hline
100 & 100 & 30\% & 60\% & 59.9 & 58.5 & 73.0 & 56.1 & 50.3 & 52.2 & 50.6 & 72.7 & 53.6 & 46.9 & 74.7 & 67.0 & 65.2 & 73.0 & 76.2 \\
100 & 200 & 30\% & 60\% & 63.9 & 65.3 & 78.8 & 62.4 & 56.6 & 54.5 & 54.4 & 78.3 & 59.4 & 51.4 & 79.6 & 72.9 & 67.3 & 75.2 & 78.7 \\
200 & 200 & 30\% & 60\% & 69.2 & 73.1 & 84.2 & 69.1 & 64.5 & 58.3 & 60.9 & 83.8 & 65.5 & 58.6 & 80.7 & 76.1 & 66.8 & 76.0 & 80.8 \\
250 & 500 & 30\% & 60\% & 78.0 & 86.5 & 92.6 & 80.3 & 79.7 & 65.0 & 78.7 & 92.2 & 76.7 & 75.2 & 88.0 & 83.7 & 71.6 & 81.3 & 85.2 \\
500 & 500 & 30\% & 60\% & 82.8 & 92.1 & 96.1 & 86.6 & 87.6 & 69.6 & 86.9 & 95.9 & 83.1 & 83.7 & 91.6 & 86.9 & 74.1 & 83.9 & 87.4 \\
\hline
\end{tabular}
\end{adjustbox}
\caption{Empirical power and coverage probabilities in percent under $H_1$ based on piecewise exponential distributions.}
\label{pwexp_h1}
\end{table}

\subsection{Practical recommendations}
Based on our simulation results and the theoretical findings of Section \ref{sec:corrections_typeI_errors} we make the following recommendations:
\begin{itemize}
    \item The methods based on the modified delta methods for the bootstrap of \citet{fang2019inference} paired with the cloglog-transformation as well as the numerical delta method with bias correction of \citet{hong2018numerical} lead to tests which are able to control the specified type-I error while maintaining good power in many settings.
    \item Regarding confidence intervals, the situation is two-fold: as long as the true effect is not too close to 0, the interval $\text{CI}^{N(2)}$ works well, often also in smaller samples. If the true effect is very close to 0, obtaining a two-sided confidence interval with reasonable coverage seems to be very challenging. We therefore recommend to only compute right-sided intervals with the left boundary set to 0, or to treat the two-sided intervals as merely illustrative and exploratory, avoiding any inferential statements based on them if the point estimate $\hat{\Delta}_{\tau}(S_1, S_2)$ is very low.
    \item To achieve reasonable type-I errors, relatively large sample sizes are necessary, especially if higher censoring proportions are expected.
    \item In situations where the proportional hazards assumptions is reasonable, a test tailored to this setup such as the proposal of \citet{wellek1993} is expected to achieve considerably better type-I error control, especially since this setup proved to be the most challenging for our methods in our simulations.
\end{itemize}

\section{Real-Data Application}

In order to evaluate the practical relevance of our method, we consider the OAM4971g study (METLung, \citealp{spigel2017}). The METLung trial was a phase III study investigating whether the addition of the MET-targeting antibody onartuzumab to erlotinib could improve the outcomes in previously treated patients with advanced MET-positive non-small cell lung cancer. Earlier phase II data had suggested a potential benefit, making MET inhibition a promising therapeutic strategy. The data is not publicly available but have been digitized and published in the R package \texttt{kmdata} \citep{fell2021kmdata}. 

In this trial, a total of 499 patients were randomized at a one-to-one ratio to receive either erlotinib plus onartuzumab or erlotinib plus placebo. The primary endpoint was overall survival (OS), with progression-free survival (PFS), response rate, and safety assessed as secondary outcomes. Figure \ref{fig:casestudy} shows the Kaplan–Meier curves for OS and PFS by treatment arm. 
Although the study was designed as a superiority trial rather than an equivalence or non-inferiority trial, such a post hoc analysis may still provide clinically meaningful conclusions. 

\begin{figure}[htp]
\centering
\includegraphics[width=.49\textwidth]{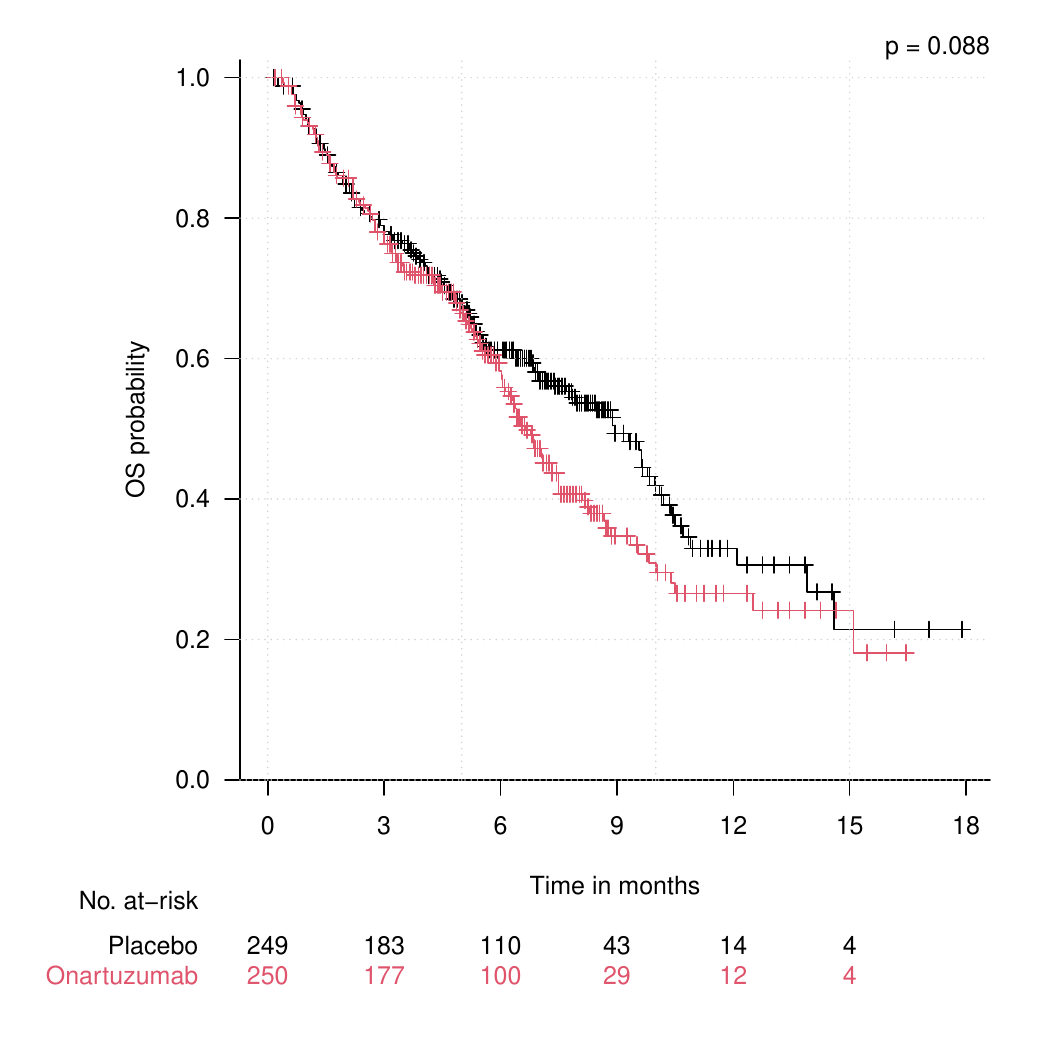}\hfill
\includegraphics[width=.49\textwidth]{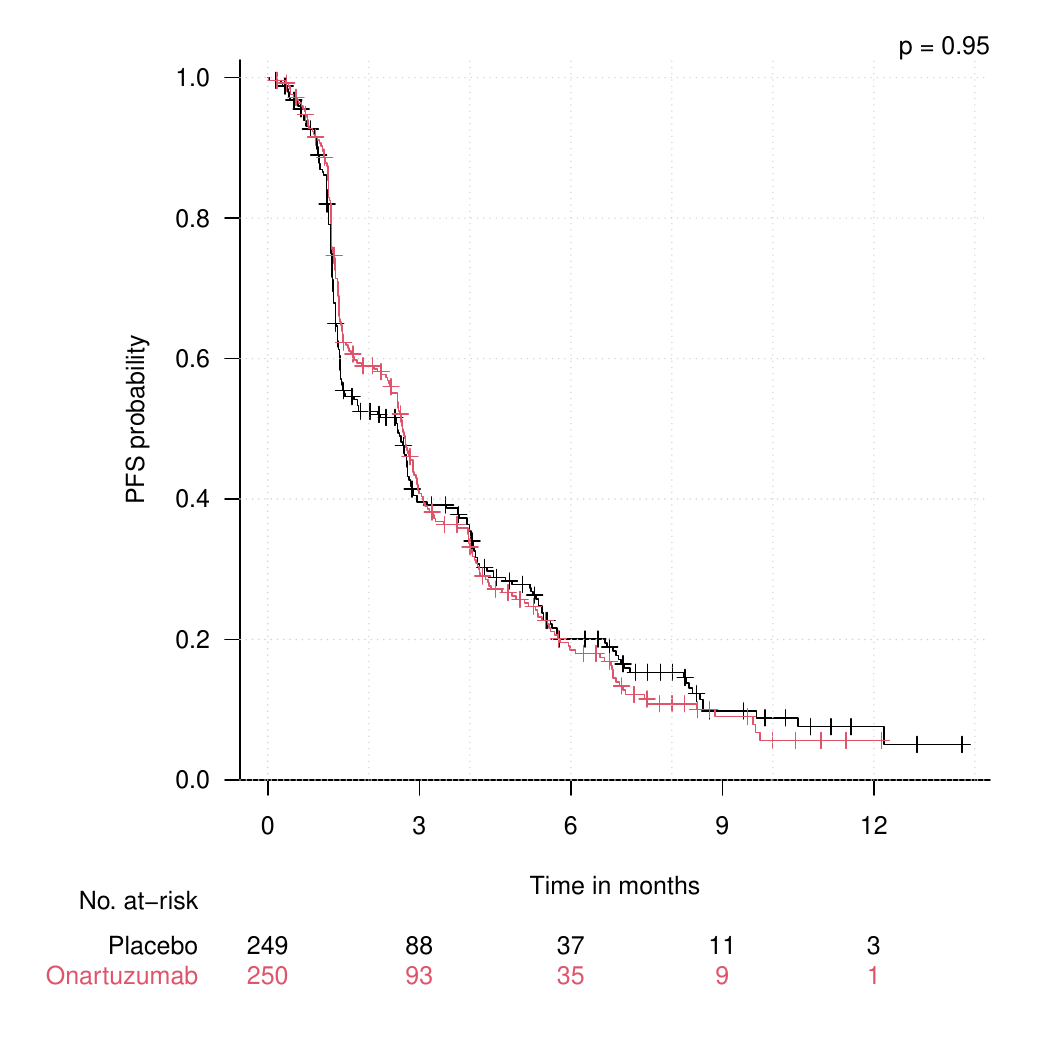}
\caption{Kaplan–Meier curves and numbers of patients at risk at several time points for overall survival (OS) and progression-free survival (PFS) by treatment arm. The $p$-values displayed in the top right correspond to log-rank tests between treatment groups.}
\label{fig:casestudy}
\end{figure}

Regarding superiority, the study failed to demonstrate a clinical benefit for the combination therapy. Median overall survival and progression-free survival were not improved in the onartuzumab arm, given by 8.9 (2.6) months in the control group vs. 6.7 (2.7) months in the onartuzumab arm for OS (PFS, respectively), and the log rank tests could not conclude any significant differences between the treatment groups ($p$-value $=0.088$ for OS, $p$-value $=0.95$ for PFS, see Figure \ref{fig:casestudy}). Further, both figures indicate the non-proportionality of hazard rates, which can substantially affect the interpretation and performance of the log-rank tests. 

However, these non-significant results alone do not prove that there is no relevant difference between the two treatments, which might be another clinical question as safety concerns (e.g., higher rates of adverse events) or additional costs would result in favoring the placebo treatment in case of similar efficacy. In order to answer this question, we consider the equivalence hypothesis as derived in \eqref{equivalencehypothesis}, utilizing different equivalence thresholds $\varepsilon$.

Figure \ref{fig:casestudy2} displays the resulting $p$-values of the tests in dependence of the equivalence threshold $\varepsilon$ for OS and PFS, i.e., using the delta method and the different subsampling and bootstrap procedures, respectively, to obtain critical values.  
Calculating multiple $p$-values for different thresholds can also be justified mathematically; details are provided in Appendix \ref{appendix:dynamic}.
As the observation period ended after 18 months, we consider $\tau=18$.
In view of our simulation results, the Fang \& Santos bootstrap method resulted in acceptable type-I error probabilities in the piece-wise exponential model setting with partly overlapping survival curves which seems to describe the present lung cancer data most accurately. For total sample sizes of about $n_1+n_2=500$ these type-I error probabilities are close to or below 0.05, so we can also assume the type-I error to be controlled in the present analysis.
Regarding OS, the point estimate equals $\hat{\Delta}_{\tau}(S_1, S_2)=0.054$, which is displayed as the red dashed line in the left panel of Figure~\ref{fig:casestudy2}. 
We observe that the smallest values of $\varepsilon$ for which $H_0$ can be rejected are given by $0.038$, $0.05$, $0.06$, $0.07$ and $0.052$ (for Fang \& Santos, delta method, adj. delta method, Efron's bootstrap and subsampling, respectively).

Similarly, regarding PFS, the test point estimate equals $\hat{\Delta}_{\tau}(S_1, S_2)=0.0185$, which is much smaller, as the survival curves are very similar (see Figure \ref{fig:casestudy}). The obtained minimal thresholds for rejection are here given as 0.006, 0.012, 0.016, 0.020 and 0.004 (same order as above). In general, $p$-values are comparably small, which is a direct consequence of the test statistic's very small value.

Finally, we note that selecting the equivalence threshold generally involves a clinical decision made in advance. The threshold should represent a clinically negligible difference and, at the same time, account for the study’s ability to detect equivalence with sufficient reliability.

\begin{figure}[htp]
\centering
\includegraphics[width=.49\textwidth]{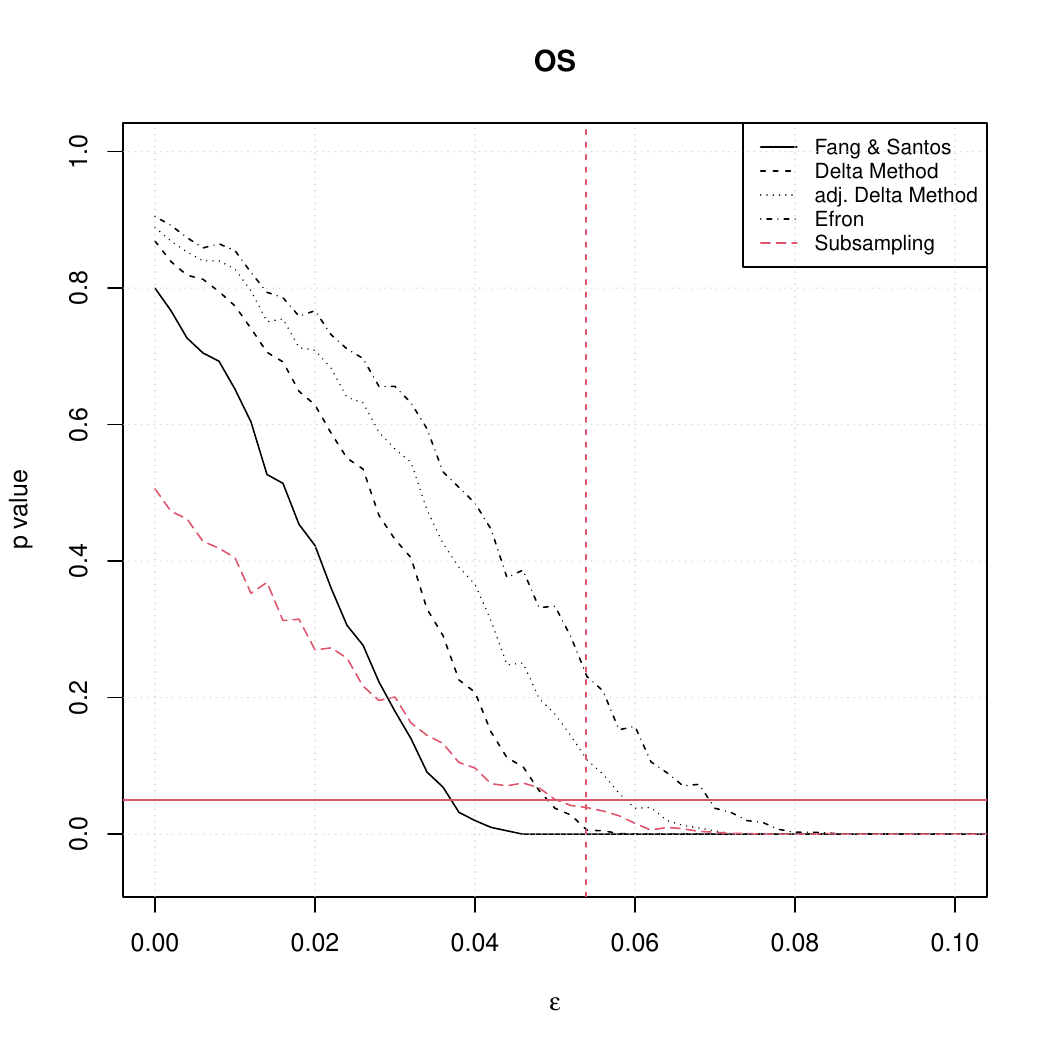}\hfill
\includegraphics[width=.49\textwidth]{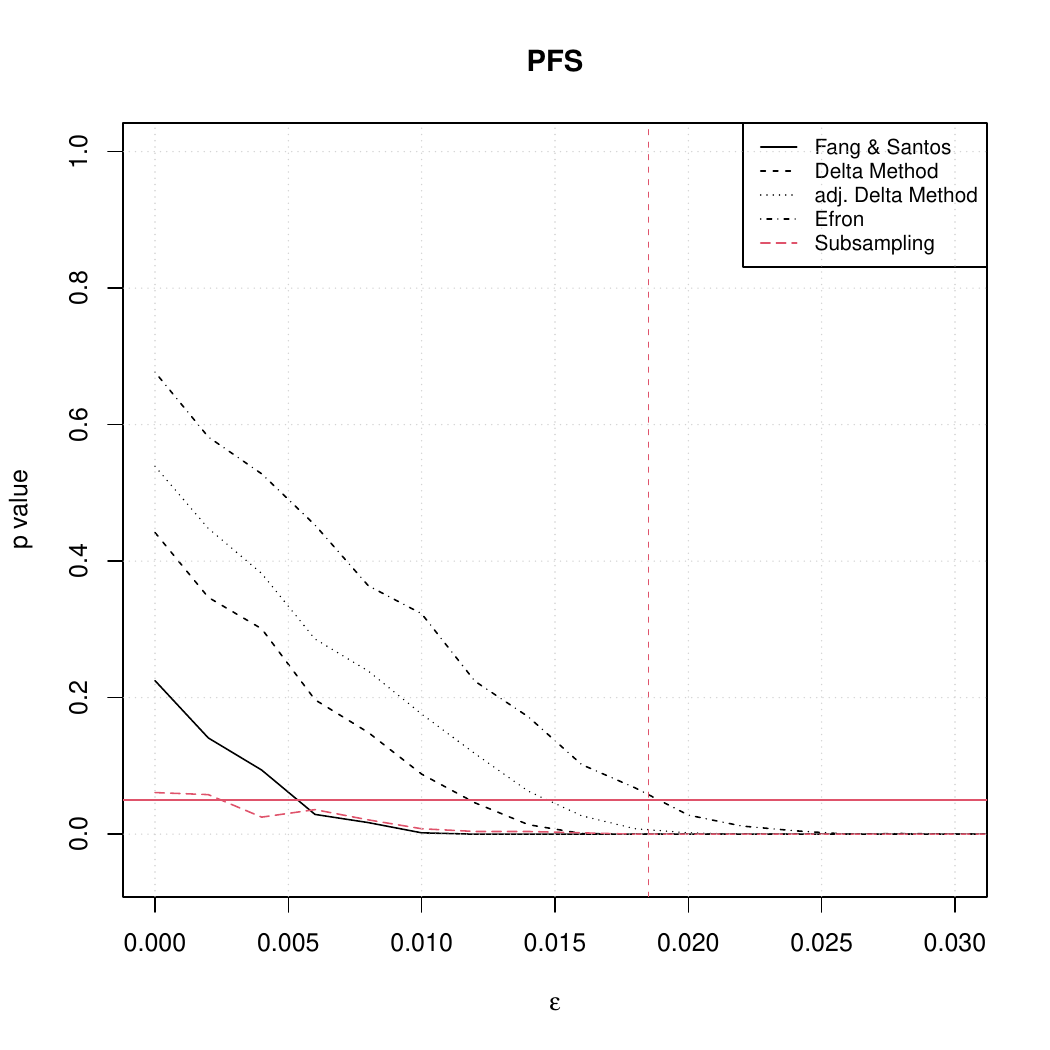}
\caption{Obtained $p$-values of the tests in dependence of the equivalence threshold $\varepsilon$ for the delta method and the different subsampling and bootstrap procedures, respectively, for OS and PFS. The solid red horizontal line indicates the $5\%$-significance level, the vertical dashed red lines indicates the value of the test statistic.}
\label{fig:casestudy2}
\end{figure}

\section{Conclusion and Discussion}
We have considered the two-sample testing problem of equivalence of two survival functions from independent groups. As our theoretical findings and simulation results demonstrate, in full generality this inference problem is difficult, in the sense that larger samples and more complicated methods are needed in comparison to classical two-sample tests for equality of distributions or for other estimands relevant to survival analysis. Nonetheless, our results indicate that the $L^1$-distance constitutes a powerful tool to quantify the difference as well as similarity of survival functions under minimal assumptions. Even though both applications are handled in a joint theoretical framework, they should be distinguished in applications: We recommend using the proposed effect measure either to quantify differences between survival curves by means of a confidence interval or to test for their equivalence, but not both in the same study because these two approaches address different statistical problems.

Future research could consider a full analysis of $L^1$-distances between cumulative incidence functions in the case of competing risks or  $L^1$-distances between transition probability curves in more general multi-state models.
In the supplementary materials for the present paper the procedure in the competing risks case is discussed in some detail.
With a glance at the present simulation results in the survival case, however, we expect that much larger sample sizes will be required to control the type-I error of equivalence tests in the competing risks case.

\section*{Acknowledgements}

We wish to thank Zoe Lange (Ruhr University Bochum) and Markus Pauly (TU Dortmund University) for helpful discussions.

\newpage
\appendix
\section*{Appendix}

The appendices are organized as follows: Appendix \ref{appendix:proofs} contains the proofs of the results from the main part, as well as additional mathematical results that are necessary for these proofs. In Appendix \ref{appendix:comprisk} we extend our results to the competing risks case. Appendix \ref{appendix:dynamic} contains a discussion regarding the data dependent choice of the equivalence threshold $\varepsilon.$

\section{Proofs}\label{appendix:proofs}
In this section, the proofs of all stated results are presented. As in the main body, we denote by $D[0,\tau]$ the space of càdlàg functions equipped with the supremum norm. We first need a simple continuity result, which we prove for the sake of completeness.
\begin{lemma}\label{continous_integral}
    The mapping $\Psi(f)=\frac{1}{\tau}\int_0^{\tau}|f(t)|dt$ from $D[0,\tau]$ to $\mathbb{R}$ is continuous.
\end{lemma}

\begin{proof}
    Let $(f_n)_{n\in\mathbb{N}}$ be a sequence in $D[0,\tau]$ converging uniformly to $f\in D[0,\tau]$ as $n\to \infty.$ Then by the (reverse) triangle- and Hölder inequality
    \begin{equation*}
    \begin{split}
                &\left|\frac{1}{\tau}\int_0^{\tau}|f_n(t)|dt-\frac{1}{\tau}\int_0^{\tau}|f(t)|dt\right|\\
                =&\frac{1}{\tau}\left|\int_0^{\tau}|f_n(t)|-|f(t)|dt\right|\\
                \leq&\frac{1}{\tau}\int_0^{\tau}\left|f_n(t)-f(t)\right|dt\\
                \leq& \sup_{t\in\tau}|f_n(t)-f(t)|
    \end{split}
    \end{equation*}
    which converges to 0 as $n\to\infty.$
\end{proof}

With this result, we can prove the consistency of $ \hat{\Delta}_{\tau}(S_1, S_2).$
\begin{proof}[Proof of Theorem~\ref{consistency_effectmeasure}]
    The Kaplan-Meier estimators $\hat{S}_j$ are elements of $D[0,\tau]$ and as noted e.g. in Lemma 1 of \citet{liu2020resampling}, they are continuous functionals of empirical processes indexed by VC-classes. Because every VC-class is also a Donsker and therefore Glivenko-Cantelli class (see e.g. Chapter 2.6 of \citealp{vanderVaartWellner2023}), the vector $(\hat{S}_1, \hat{S}_2)$ converges outer almost surely to $(S_1, S_2).$ The result now follows from the continuous mapping theorem (Theorem 1.11.1 of \citealp{vanderVaartWellner2023}), Lemma \ref{continous_integral} and the continuity of addition on $D[0,\tau].$
\end{proof}

In order to establish our weak convergence results we need another preliminary analytic result, the directional differentiability of $\Psi.$ A similar result was already established by \citet{bastian2024comparing}, but they considered the functional $\Psi$ on a different function space. 

\begin{lemma}\label{directional_diff_L1}
    The functional $\Psi(f)$ is directionally Hadamard differentiable at $f\in D[0,\tau]$ with derivative $\Psi'_{f}(h)=\frac{1}{\tau}\int_{\{f=0\}}|h(x)|dx+\frac{1}{\tau}\int_{\{f\ne0\}}\operatorname{sign}(f(x))h(x)dx.$
\end{lemma}
\begin{proof}
    
    The mapping $\Psi$ can be represented as the composition $\Psi_2\circ \Psi_1$ where $\Psi_1:D[0,\tau]\to L^1[0,\tau]$ is the natural injection, $\Psi_1(f)=f$, and $\Psi_2:L^1[0,\tau]\to \mathbb{R}$ is a multiple of the usual $L^1$ norm, $\Psi_2(f)=\frac{1}{\tau}||f||_{L^1}.$ The functional $\Psi_1$ is linear and by Lemma \ref{continous_integral} also continuous, which implies it equals its own Hadamard derivative.  By the chain rule (Proposition 3.6 of \citet{shapiro1990concepts}) it remains to show that 
    \begin{equation*}\label{hadamard_to_show}
        \left|\frac{\frac{1}{\tau}\int_0^{\tau}|f(x)+h_n(x)t_n|dx-\frac{1}{\tau}\int_0^{\tau}|f(x)|dx}{t_n}-\frac{1}{\tau}\left(\int_{\{f=0\}}|h(x)|dx+\int_{\{f\ne0\}}\operatorname{sign}(f(x))h(x)dx\right)\right|
    \end{equation*}
    converges to $0$ for $h_n\to h$ in $L^1$ and $t_n \searrow 0.$ Subsequently we omit the scaling factor $1/\tau$ to improve readability. We treat the cases $f=0$ and $f\ne 0$ separately. On the set $N_f :=\{f=0\}\cap [0,\tau]$ we have 
    \begin{equation*}
    \begin{split}
                \left|\frac{\int_{N_f}|h_n(x)t_n|dx}{t_n}-\int_{N_f}|h(x)|dx\right|&=\left|\frac{t_n\int_{N_f}|h_n(x)|dx}{t_n}-\int_{N_f}|h(x)|dx\right|\\
                &=\left|\int_{N_f}|h_n(x)|-|h(x)|dx\right|\\
                &\leq \int_{N_f}|h_n(x)-h(x)|dx \to 0
    \end{split}
    \end{equation*}
    by the reverse triangle inequality and monotonicity of the integral. Now define the function
    \begin{equation*}
        g_n(x)=\frac{|f(x)+h_n(x)t_n|-|f(x)|}{t_n}.
    \end{equation*}
    Convergence in $L^1$ implies converge in measure and by a subsequence argument, we can assume w.l.o.g. that $h_n\to h$ almost everywhere. On the set $\{f\ne0\}\cap [0,\tau]$ and by considering $\{f>0\}\cap [0,\tau]$ and $\{f<0\}\cap [0,\tau]$ separately, it is easy to see that $g_n(x)$ converges almost everywhere to $g(x)=\operatorname{sign}(f(x))h(x).$ Additionally we have
    \begin{equation*}
    \begin{split}
      |g_n(x)-g(x)|&=  \left|\frac{|f(x)+h_n(x)t_n|-|f(x)|}{t_n}-\operatorname{sign}(f(x))h(x)\right|\\
      &\leq \left|\frac{|f(x)+h_n(x)t_n|-|f(x)|}{t_n}\right| + |h(x)| \\
      &\leq \frac{|f(x)+h_n(x)t_n-f(x)|}{t_n} + |h(x)|\\
      &=|h_n(x)|+|h(x)|\\
      &\leq \sup_n |h_n(x)| + |h(x)|.
    \end{split}
    \end{equation*}
    Because $L^1$ convergence implies uniform integrability of the sequence $(h_n)_{n\in\mathbb{N}},$ the function 
    $$x \mapsto \sup_n |h_n(x)| + |h(x)|$$
    is integrable. The result now follows from the dominated convergence theorem.
\end{proof}
\begin{proof}{Proof of Theorem~\ref{limitABCgeneral}}
    It follows analogously to the proofs of \citet{liu2020resampling}, that
    \[\sqrt{n}((\hat{S}_1-\hat{S}_2)- (S_1-S_2))\xrightarrow{d}\mathbb{G}\]
    on $D[0,\tau].$ The result now follows from Lemma \ref{directional_diff_L1} and an application of the functional Delta method for directionally Hadamard-differentiable functionals (Theorem 2.1 of \citet{shapiro1991asymptotic}).
\end{proof}
In order to unify the treatment of all our testing procedures, we subsequently establish a general result regarding the asymptotic behavior of hypothesis tests. It can be viewed as a formalization of some of the introductory remarks of \citet{beran1988prepivoting}. None of these results are new by themselves, and we do not use the most general possible formulation, which has the advantage of making the assumptions easy to check in applications. The essence of the lemma is that regarding asymptotic size and power, ``all hypothesis tests are created equal''. The result is only formulated for the one-sided case, a two-sided version is immediate.
\begin{lemma}\label{behavior_general_tests}
    Let $T_n$ be a sequence of real-valued statistics with
    \[r_n(T_n-\theta)\xrightarrow{d}T\]
    as $n\to\infty$ for a sequence of constants $r_n\to\infty$ as $n\to \infty$ and a non-degenerate random variable $T$ with continuous distribution function $F.$ Denote by $F_n$ a sequence of (possibly random) distribution functions with $\sup_{x\in\mathbb{R}}|F_n(x)-F(x)|\xrightarrow{p}0$ as $n \to \infty.$
    For $\alpha\in(0,1)$ and $\theta_0\in\mathbb{R}$ define \[\varphi_{n,\alpha}=\mathbbm{1}\{r_n(T_n-\theta_0)\leq  F_n^{-1}(\alpha)\}\] where $F_n^{-1}(\alpha)$ denotes the (possibly random) $\alpha$-quantile of $F_n.$ Then as $n\to \infty$
    \begin{enumerate}
        \item[(a)] $E(\varphi_{n,\alpha})\to\alpha$ if $\theta-\theta_0=0.$
        \item[(b)] $E(\varphi_{n,\alpha})\to0$ if $\theta-\theta_0>0.$
        \item[(c)] $E(\varphi_{n,\alpha})\to1$ if $\theta-\theta_0<0.$
    \end{enumerate}
    These results also remain valid if $F_n\to F$ in outer probability.
\end{lemma}
\begin{proof}[Proof of Lemma \ref{behavior_general_tests}]
    The proof of (a) is almost identical to the proof of Lemma 23.3 of \citet{vanderVaart1998}, we present the argument for sake of completeness. Assume for now, that $F^{-1}$ is continuous at $\alpha.$ Because convergence of distribution functions implies convergence of quantiles at  continuity points, an application of Slutzky's lemma and the assumption $\theta=\theta_0$ yields
    \[(r_n(T_n-\theta_0), F_n^{-1}(\alpha))\xrightarrow{d} (T, F^{-1}(\alpha))\]
    and therefore
    \[E(\varphi_{n,\alpha})=P(r_n(T_n-\theta_0)\leq  F_n^{-1}(\alpha))\to F(F^{-1}(\alpha))=\alpha\]
    as $n\to \infty.$ As quantile functions have at most countably many discontinuities, this must be valid for all but countably many $\alpha.$ Both left and right hand side of the above equation are monotone in $\alpha$ and the right hand side is continuous, therefore the statement must be valid for all $\alpha\in(0,1).$ In the case $F_n\to F$ in outer probability, the argument goes through unchanged, if the expectation $E(\varphi_{n,\alpha})$ is replaced by the outer expectation $E^{\ast}(\varphi_{n,\alpha}),$ because Slutzky's lemma is also valid for convergence in outer probability to constants (Example 1.4.7 of \citet{vanderVaartWellner2023}).

    To prove (b) and (c), note that
    \[r_n(T_n-\theta_0)=r_n(T_n-\theta)+r_n(\theta-\theta_0).\]
    The first term converges in distribution to $T,$ the second tends to (minus) infinity. Slutzky's lemma yields $r_n(T_n-\theta)\to \infty$ (b) or $r_n(T_n-\theta)\to -\infty$ (c) in (outer) probability. Because  $F_n^{-1}(\alpha)$ is asymptotically tight, the limit of $E(\varphi_{n,\alpha})$ is 0 or 1 respectively. Note that even if $F^{-1}$ is not continuous at $\alpha$,
    \[\lim_{n\to\infty} P_{\ast}(F_n^{-1}(\alpha) \in [q(\alpha)^--\delta,q(\alpha)^++\delta])=1\]
    with $q(\alpha)^-\coloneqq F^{-1}(\alpha)\leq q(\alpha)^+\coloneqq \inf\{p: F(p)>\alpha\}$ and for some $\delta>0,$ see Theorem 1 of \citet{beran1984bootstrap} for details.
\end{proof}

\begin{proof}[Proof of Theorem \ref{oracletest_limts}]
    The result follows directly from Theorem \ref{limitABCgeneral} and Lemma \ref{behavior_general_tests} because $D_{\Delta_{\tau}},$ as a mixture of two continuous distributions, has a continuous distribution function.
\end{proof}
\begin{proof}[Proof of Theorem \ref{subsampletest_limits}]
    Due to our regularity assumptions, Theorem 3.1 of \citet{politis2010k} ensures that the subsampling distribution $L_{n,b}$ converges in probability to the distribution of $D_{\Delta_{\tau}}$ under $H_0$ as well as $H_1.$ We can therefore apply Lemma \ref{behavior_general_tests}, which completes the proof.
\end{proof}
In order to prove the validity of our bootstrap procedure, we first establish a consistency result for the estimated derivative $\hat{\Psi}_n'.$
\begin{lemma}\label{consistency_derivative}
    For any fixed $h\in D[0,\tau]$ we have $\hat{\Psi}'_n(h)\to\Psi'_{S_1-S_2}(h)$ as $n\to\infty$ in (outer) probability.
\end{lemma}
\begin{proof}[Proof of Lemma \ref{consistency_derivative}]
    The subsequent argument shares the initial steps with the proof of Lemma 1 in \citet{bastian2024comparing}.

By the triangle inequality, $|\hat{\Psi}'_n(h)-\Psi'_f(h)|\leq A+B$ with
\begin{equation*}
    \begin{split}
    A=&\Bigg|\frac{1}{\tau}\int_{\{t\in[0,\tau]:|\hat{S}_1(t)-\hat{S}_2(t)|> \frac{1}{c_n} \}}\operatorname{sign}(\hat{S}_1(t)- \hat{S}_2(t))h(t)dt\\
    &-\frac{1}{\tau}\int_{\{t\in[0,\tau]:|S_1(t)-S_2(t)|> 0 \}}\operatorname{sign}(S_1(t)-S_2(t))h(t)dt\Bigg|,\\
        B=&\Bigg|\frac{1}{\tau}\int_{\{t\in[0,\tau]:|\hat{S}_1(t)-\hat{S}_2(t)|\leq \frac{1}{c_n} \}}|h(t)|dt-\frac{1}{\tau}\int_{\{t\in[0,\tau]:|S_1(t)-S_2(t)|=0 \}}|h(t)|dt\Bigg|,
    \end{split}
\end{equation*}
and it therefore suffices to show $A\to0$ and $B\to0$ in outer probability. Both terms can be handled similarly and, because \citet{bastian2024comparing} focused on $A$, we only prove the convergence of $B.$ Define the sets 
\begin{equation*}
    M_1=\{t\in[0,\tau]:|\hat{S}_1(t)-\hat{S}_2(t)|\leq 1/c_n\},\quad M_2=\{t\in[0,\tau]:|S_1(t)-S_2(t)|=0\}.
\end{equation*}
In order to avoid additional notation, we subsequently replace all sets and functions by a measurable envelope whenever necessary, without explicit distinction between them. Arguments given in \citet{vanderVaartWellner2023} formally justify this approach, we refer especially to Part I, Lemmas 1.2.1 and 1.2.2.
Due to the inequality 
\begin{equation*}
    \left|\int_{M_1}|h(x)|dx-\int_{M_2}|h(x)|dx\right|\leq\int_{M_1\oplus M_2}|h(x)|dx
\end{equation*}
and the boundedness of $h$, it is sufficient to prove 
\begin{equation}\label{convergence_symmdiff}
    \lambda(M_1\oplus M_2)\to 0
\end{equation} in (outer) probability as $n\to \infty$, where $\oplus$ denotes the symmetric difference of sets and $\lambda$ Lebesgue measure. We denote by $(\Omega, \mathcal{A}, P)$ the probability space our estimators are based upon. Additionally denote by $\mathcal{B}([0,\tau])$ the Borel $\sigma$-field on $[0,\tau].$ If we define a (probability) measure on the measurable space $([0,\tau]\times\Omega, \mathcal{B}([0,\tau])\times \mathcal{A})$ by
$\Tilde{P}=\Tilde{\lambda}\otimes P$ with $\Tilde{\lambda}=\frac{\lambda_{|[0,\tau]}}{\tau}$, we can reformulate \eqref{convergence_symmdiff} as
\[\mathbbm{1}\{E_n\cup Q_n\}=:u_n\xrightarrow{\Tilde{P}}0\]
with $$E_n=\{t\in[0,\tau]:|\hat{S}_1(t)-\hat{S}_2(t)|\leq 1/c_n, |S_1(t)-S_2(t)|>0 \} \subset [0,\tau] \times \Omega,$$
and $$Q_n=\{t\in[0,\tau]:|\hat{S}_1(t)-\hat{S}_2(t)|> 1/c_n, |S_1(t)-S_2(t)|=0 \} \subset [0,\tau] \times \Omega.$$ Due to the uniform consistency of the Kaplan-Meier estimators, for any $\eta>0$ and any $\delta>0$ there exists an $n_0 = n_0(\eta, \delta)$ which does not depend on $t$ such that the event $H_j=\{t\in[0,\tau]:|\hat{S}_j(t)-S_j(t)|\leq\delta\}$ has $\Tilde{P}$-probability at least $1-\eta/2$ for $n\geq n_0.$ Also note that 
\begin{equation*}
    E_n\cap H_1\cap H_2 \subset E_n'=\{t\in[0,\tau]:0<|S_1(t)-S_2(t)|\leq \delta +1/c_n\},
\end{equation*}
since
\begin{equation*}
    |S_1(t)-S_2(t)|\leq |\hat{S}_1(t)-\hat{S}_2(t)|+ |\hat{S}_1(t)-S_1(t)|+|\hat{S}_2(t)-S_2(t)|.
\end{equation*}
Additionally 
\begin{equation*}
    Q_n\subset Q_n'=\{t\in[0,\tau]:|(\hat{S}_1(t)-\hat{S}_2(t))- (S_1(t)-S_2(t))|> 1/c_n\}.
\end{equation*}
Therefore, for $n\geq n_0$ and by Fubini's theorem
\begin{equation*}
    \begin{split}
        &\int u_nd\Tilde{P}\\
        =&\Tilde{P}(E_n+Q_n)\\
        \leq&\Tilde{P}(E_n)+\Tilde{P}(Q_n)\\
        =&\int\int\mathbbm{1}\{E_n\}d\Tilde{\lambda} dP+ \int\int\mathbbm{1}\{Q_n\}d\Tilde{\lambda} dP\\
        \leq& \int\int\mathbbm{1}\{E_n'\}d\Tilde{\lambda} dP +\eta + \int\int\mathbbm{1}\{Q_n'\}d\Tilde{\lambda} dP\\
        \leq& \int\mathbbm{1}\{E_n'\}d\Tilde{\lambda} +\eta + P(\sup_{t\in[0,\tau]}\sqrt{n}|(\hat{S}_1(t)-\hat{S}_2(t))- (S_1(t)-S_2(t))|> \sqrt{n}/c_n).
    \end{split}
\end{equation*}
By letting $n\to\infty$ followed by $\delta\searrow0$ the first and third term vanish by continuity from above of the Lebesgue measure and the weak convergence of the Kaplan-Meier estimators on $D[0,\tau]$ together with assumption \eqref{boot_ratecondition}, respectively. This completes the proof, as $\eta>0$ was arbitrary.
\end{proof}

\begin{proof}[Proof of Theorem \ref{boottest_limits}]
     For $h_1,h_2\in D[0,\tau]$ by direct calculation
\[|\hat{\Psi}'_n(h_1)-\hat{\Psi}'_n(h_2)|\leq 2\sup_{t\in[0,\tau]}|h_1(t)-h_2(t)|,\]
i.e. $\hat{\Psi}'_n(h)$ is Lipschitz-continuous in $h.$ Due to Remark 3.4 of \citet{fang2019inference} and Lemma \ref{consistency_derivative}
we can apply Theorem 3.2 of \citet{fang2019inference}, which yields
    \begin{equation*}
    \hat{\Psi}'_n(\sqrt{n}((\hat{S}_1^{(B)}-\hat{S}_2^{(B)}) -(\hat{S}_1-\hat{S}_2))) \xrightarrow{d} D_{\Delta_{\tau}}
\end{equation*}
given the data in (outer) probability. Applying Lemma \ref{behavior_general_tests} completes the proof.
\end{proof}

\section{Extension to competing risk models}\label{appendix:comprisk}
Our approaches can easily be extended to a two-sample competing risks setup. To do so, we need to introduce some additional notation. Each sample, $j=1,2$, is modeled by individual competing risk processes  $(Z_{j}(t))_{t\ge0}$ with $m\in \mathbb{N}$ competing risks. We assume these processes to be càdlàg,  non-homogeneous Markov processes with state space $\{0,1, \dots, m\}$ and initial state $0$, i.e., $P(Z_{j}(0) = 0) = 1$. The other states $1, \dots, m$ represent the competing events and are absorbing.  For ease of presentation, we limit ourselves to the consideration of only $m=2$ risks. In most applications, this can be done without loss of valuable information if only one competing risk (here, the first) is of interest; all other causes of failure can be combined into a single competing endpoint (the second), which is then treated as a nuisance parameter.
The event time now is defined as the transition time out of the initial state $T_j=\inf\{t>0 | Z_{j}(t) \ne 0\}$, and we denote by $D_{j} \in \{1,2\}$ the event type. Furthermore, denote by $N_{jk}(t)=\sum_{i=1}^{n_j}\mathbbm{1}\{X_{ij}\leq t, \delta_{ij}=1, D_{ij}=k\}, k=1,2$ the number of observed cause-$k$ events up to $t$ in group $j.$
The distribution of $T_j$ is determined by the cumulative cause-specific hazard functions
\begin{equation*}
    A_{jk}(t)=\int_0^t \frac{1}{S_j(s-)}dH_{jk}(s), k=1,2
\end{equation*}
with $H_{jk}(s)=E(N_{jk}(s)),$ as the transition probability matrix of the process $Z_j$ can be recovered from them in terms of the product integral, c.f. \citet{gill1990survey}.
The main quantities of interest in the competing risks setup are the cumulative incidence functions (CIFs), also called subdistribution functions
\begin{equation*}
	F_{jk}(t) = P(T_{jk} \le t, D_{j}=k) = \int_{0}^{t}S(s-)dA_{jk}(s).
\end{equation*}
Since we are only interested in the first risk, we define our generalized discrepancy measure as
\begin{equation}\label{effectmeasure_comprisks}
    \Delta_{\tau}(F_{11}, F_{21})= \frac{1}{\tau}\int_0^{\tau}|F_{11}(t)-F_{21}(t)|dt,
\end{equation}
which also is contained in $[0,1].$ However, the upper bound is in general not attainable outside the case that the second risk only occurs with probability 0. While rescaling by e.g. $1/\max(F_{11}(\tau), F_{21}(\tau))$ would lead to an attainable upper bound, we refrain from doing so, as the measure would lose much of its geometric interpretation.
The hypothesis of equivalence in the competing risks setting can be formulated in terms of the generalized discrepancy measure \eqref{effectmeasure_comprisks} as
\begin{equation}\label{equivalencehypothesis_comprisk}
    H_0:\Delta_{\tau}(F_{11}, F_{21}) \geq \varepsilon \quad \text{versus}\quad H_1:\Delta_{\tau}(F_{11}, F_{21})<\varepsilon.
\end{equation}
We estimate the CIFs by the Aalen-Johansen estimators \citep{aalen1978empirical}
 \begin{equation*} \label{AJest}
 	\hat{F}_{jk}(t) = \int_{0}^{t} \frac{\hat{S}_j(s-)}{Y_j(s)} dN_{jk}(s),  
 \end{equation*}
which leads to the plug-in estimator for the generalized discrepancy measure \eqref{effectmeasure_comprisks}
\begin{equation*}\label{estimatedL1_distance_comprisk}
    \hat{\Delta}_{\tau}(F_{11}, F_{21})= \frac{1}{\tau}\int_0^{\tau}|\hat{F}_{11}(t)-\hat{F}_{21}(t)|dt.
\end{equation*}
The subsequent results are stated without formal proof, they can be established analogously to the corresponding results of the main paper, if the (conditional) convergence results for Aalen-Johansen estimators of \citet{dobler2017discontinuity, dobler2024erratum} are used instead of those for the Kaplan-Meier estimator. The following theorem establishes consistency and the asymptotic distribution of the \eqref{effectmeasure_comprisks}.
\begin{theorem}\label{limit_comprisk}
Suppose assumption \eqref{groupassumption} holds. Then as $\min(n_1,n_2)\to\infty,$ we have
    \begin{itemize}
        \item[(a)]
        \begin{equation*}
            \hat{\Delta}_{\tau}(F_{11}, F_{21})\to \Delta_{\tau}(F_{11}, F_{21})
        \end{equation*}
        in (outer) probability.
        \item[(b)]
        \begin{equation*}
            \begin{split}
        \sqrt{n}(\hat{\Delta}_{\tau}(F_{11}, F_{21})-\Delta_{\tau}(F_{11}, F_{21}))\xrightarrow{d}&\Tilde{D}_{\Delta_\tau}\coloneqq\frac{1}{\tau}\int_{\{t\in[0,\tau]:F_{11}(t)=F_{21}(t) \}}|\Tilde{\mathbb{G}}(t)|dt\\
                    &+
        \frac{1}{\tau}\int_{\{t\in[0,\tau]:F_{11}(t)\ne F_{21}(t)\}}\operatorname{sign}(F_{11}(t)-F_{21}(t))\Tilde{\mathbb{G}}(t)dt
    \end{split}
        \end{equation*}
    \end{itemize}
    where $\Tilde{\mathbb{G}}$ is a mean zero Gaussian process with an involved covariance function which can be derived from the one-sample results of \citet{dobler2017discontinuity, dobler2024erratum}.
\end{theorem}
\begin{remark}
    Theorem 5.1 of \citet{dobler2017discontinuity} only allows for a finite number of jumps of the cumulative incidence function. However, this condition is generally not restrictive in practice, as tied event times in clinical trials are usually a result of discretizing the time scale, i.e. rounding to full calender days.
\end{remark}
If we denote the $\alpha$-quantile of $\Tilde{D}_{\Delta_\tau}$ by $\Tilde{q}_{\alpha},$ a test for the equivalence hypothesis \eqref{equivalencehypothesis_comprisk} can be defined analogously to \eqref{oracletest}:

\begin{equation*}
    \Tilde{\varphi}_{n,\alpha}=\mathbbm{1}\left\{\hat{\Delta}_{\tau}(F_{11}, F_{21})-\varepsilon\leq \frac{\Tilde{q}_{\alpha}}{\sqrt{n}}\right\}.
\end{equation*}
Similar to the classical survival case, resampling procedures are necessary to obtain a feasible test, as the mixture weights in the distribution of $\Tilde{D}_{\Delta_\tau}$ as well as the covariance function of $\Tilde{\mathbb{G}}$ are unknown. We only present the subsampling and modified bootstrap analogues of sections \ref{sec:subsample} and \ref{sec:modified_boot}, the adjustments described in section \ref{sec:corrections_typeI_errors} can be defined analogously.
Let $\Tilde{Z}_{1j},\dots,\Tilde{Z}_{N_j}$ be equal to the $N_j=\binom{n_j}{b_j}$ possible unordered subsets of $(X_{1j}, \delta_{1j}, D_{ij}), \dots, (X_{n_jj}, \delta_{n_jj}, D_{n_jj}).$ Denote by $\hat{F}_{l_jj1}$ the Aalen-Johansen estimator computed from the subset $Z_{l_j}.$
Then the subsampling distribution in the competing risks setup is given by
\begin{equation*}
    L^{CR}_{n,b}(x)=\frac{1}{\binom{n_1}{b_1}\binom{n_2}{b_2}}\sum_{l_1=1}^{\binom{n_1}{b_1}}\sum_{l_2=1}^{\binom{n_2}{b_2}}\mathbbm{1}\left\{\sqrt{b}\left(\frac{1}{\tau}\int_0^{\tau}|\hat{F}_{l_111}(t)-\hat{F}_{l_221}(t)|dt-\frac{1}{\tau}\int_0^{\tau}|\hat{F}_{11}(t)-\hat{F}_{21}(t)|dt\right)\leq x\right\}
\end{equation*}
Based on the subsampling distribution, we obtain the equivalence test
\begin{equation*}
        \Tilde{\varphi}_{n,b,\alpha}^S=\mathbbm{1}\left\{\hat{\Delta}_{\tau}(F_{11}, F_{21})-\varepsilon\leq\frac{L^{CR(-1)}_{n,b}(\alpha)}{\sqrt{n}} \right\}.
\end{equation*}
where $L^{CR(-1)}_{n,b}(\alpha)$ denotes the $\alpha$-quantile of the (random) distribution function $L^{CR}_{n,b}.$
For this test, an analogous result to Theorem \ref{subsampletest_limits} holds:
\begin{theorem}\label{subsampletest_limits_comprisk}
    Suppose assumptions \eqref{groupassumption} and \eqref{subsamplingassumption} hold. Then as $\min(n_1,n_2)\to \infty$
    \begin{itemize}
        \item[(a)] If $\Delta_{\tau}(F_{11}, F_{21})= \varepsilon$, then $\lim_{n\to\infty} E(\Tilde{\varphi}_{n,b,\alpha}^S)=\alpha.$
        \item[(b)] If $\Delta_{\tau}(F_{11}, F_{21})> \varepsilon$, then $\lim_{n\to\infty} E(\Tilde{\varphi}_{n,b,\alpha}^S)=0.$
        \item[(c)] If $\Delta_{\tau}(F_{11}, F_{21})< \varepsilon$, then $\lim_{n\to\infty} E(\Tilde{\varphi}_{n,b,\alpha}^S)=1.$
    \end{itemize}
\end{theorem}
The modified bootstrap of \citet{fang2019inference} can be generalized to the competing risks setting as well. As an estimator for the directional derivative we propose
\begin{equation*}\label{estimated_derivative_comprisk}
\begin{split}
        \hat{\Tilde{\Psi}}_n'(h)=&\frac{1}{\tau}\int_{\{t\in[0,\tau]:|\hat{F}_{11}(t)-\hat{F}_{21}(t)|\leq \frac{1}{c_n} \}}|h(t)|dt\\ &+\frac{1}{\tau}\int_{\{t\in[0,\tau]:|\hat{F}_{11}(t)-\hat{F}_{21}(t)|> \frac{1}{c_n} \}}\operatorname{sign}(\hat{F}_{11}(t)-\hat{F}_{21}(t))h(t)dt
\end{split}
\end{equation*}
Denoting by $\hat{F}^{(B)}_{j1}$ a bootstrap counterpart of $\hat{F}_{j1},$ we define the bootstrap approximation by
\begin{equation*}
    \hat{\Tilde{\Psi}}_n'(\sqrt{n}((\hat{F}^{(B)}_{11}-\hat{F}^{(B)}_{21})- (\hat{F}_{11}-\hat{F}_{21})))
\end{equation*}
and denote its conditional $\alpha$-quantile by $\Tilde{q}^B_{n,\alpha}.$
\begin{remark}
    While it is possible to use the Efron bootstrap for the Aalen-Johansen estimator, usually a wild bootstrap approach is preferable, as it leads to better finite sample behavior. See \citet{lin1997non, beyersmann2013weak, dobler2017non} for details regarding the wild bootstrap of the Aalen-Johansen estimator in the continuous case, and \citet{dobler2017discontinuity, dobler2024erratum} for necessary adjustments in the discontinuous case.
\end{remark}
Based on the conditional quantiles $\Tilde{q}^B_{n,\alpha}$ we obtain the test
\begin{equation*}
        \Tilde{\varphi}_{n,\alpha}^B=\mathbbm{1}\left\{\hat{\Delta}_{\tau}(F_{11}, F_{21})-\varepsilon\leq\frac{q^{B}_{n,\alpha}}{\sqrt{n}} \right\}
\end{equation*}
for the equivalence hypothesis \eqref{equivalencehypothesis_comprisk}. It's asymptotic properties are summarized in the theorem below.
\begin{theorem}\label{boottest_limits_comprisk}
    Suppose assumptions \eqref{groupassumption} and \eqref{boot_ratecondition} hold. Then as $\min(n_1,n_2)\to \infty$
    \begin{itemize}
        \item[(a)] If $\Delta_{\tau}(F_{11}, F_{21})= \varepsilon$, then $\lim_{n\to\infty} E(\Tilde{\varphi}_{n,\alpha}^B)=\alpha.$
        \item[(b)] If $\Delta_{\tau}(F_{11}, F_{21})>\varepsilon$, then $\lim_{n\to\infty} E(\Tilde{\varphi}_{n,\alpha}^B)=0$
        \item[(c)] If $\Delta_{\tau}(F_{11}, F_{21})< \varepsilon$, then $\lim_{n\to\infty} E(\Tilde{\varphi}_{n,\alpha}^B)=1.$
    \end{itemize}
\end{theorem}

\section{Dynamic choice of the equivalence threshold}\label{appendix:dynamic}
A general difficulty of equivalence tests is the choice of the thresholds used to constitute equivalence which also applies to our methods. Especially in the competing risks setup where range of the effect measure is usually unknown, it may be difficult to find a reasonable value for $\varepsilon$ in advance. However, it is possible to choose $\varepsilon$ in a data dependent fashion. We adopt the recent proposal of \citet{mollenhoff2024testing}, who suggested the subsequent method in the context of parametric competing risk models, see also Remark 1 (b) of \citet{lange2025testing}.

We formulate the application directly to the competing risks setting, which includes the classical survival setup considered in the main part of the paper as a special case for $F_{12} = F_{22} \equiv 0$. For any fixed $\varepsilon_1\leq \varepsilon_2$ we have 
\begin{equation*}
    \{F_{11},F_{21}:\Delta_{\tau}(F_{11}, F_{21})\geq \varepsilon_2\}\subset\{F_{11},F_{21}:\Delta_{\tau}(F_{11}, F_{21})\geq \varepsilon_1\}
\end{equation*}
i.e. the null hypotheses in \eqref{equivalencehypothesis_comprisk} are nested. Furthermore,
\begin{equation*}
    \mathbbm{1}\left\{\hat{\Delta}_{\tau}(F_{11}, F_{21})-\varepsilon_1\leq \frac{\Tilde{q}_{\alpha}}{\sqrt{n}}\right\} \leq     \mathbbm{1}\left\{\hat{\Delta}_{\tau}(F_{11}, F_{21})-\varepsilon_2\leq \frac{\Tilde{q}_{\alpha}}{\sqrt{n}}\right\}
\end{equation*}
and the same is true if the subsampling or bootstrap quantiles $\Tilde{q}^S_{n,b,\alpha}$ and $\Tilde{q}^B_{n,\alpha}$ are used. Therefore a rejection of the null hypothesis \eqref{equivalencehypothesis_comprisk} for $\varepsilon$ also implies the rejection of the null for all $\Tilde{\varepsilon}>\varepsilon.$ This allows for the application of the sequential rejection principle \citep{goeman2010sequential}, which implies that the following procedure controls the type I error asymptotically: First, choose a grid $0<\varepsilon_1<\varepsilon_2<\dots< \varepsilon_r<1$ for some $r\in \mathbb{N}.$ Then test \eqref{equivalencehypothesis_comprisk} with $\varepsilon_1$; in case the null hypothesis is not rejected, proceed with $\varepsilon_{2},\varepsilon_{3}$ etc. until the first rejection occurs. If the null hypothesis is rejected in the $j$-th step, $\varepsilon_j$ can be interpreted as a measure of evidence for similarity between $F_{11}$ and $F_{21}.$ For practical applications, we recommend to plot the different values of $\varepsilon$ (x-axis) together with the $p$-values (y-axis) from the corresponding tests against a horizontal line at $\alpha$ in a joint coordinate system, then the intersection of both lines yields the value $\varepsilon_j$ in the procedure described above.

\bibliographystyle{abbrvnat}
\bibliography{literatur}

\end{document}